\documentclass[12pt,draftcls,onecolumn,peerreview]{IEEEtran}
\usepackage{ifpdf}
 \ifpdf
 \else
 \fi

\usepackage{cite}

\ifCLASSINFOpdf
 \usepackage[pdftex]{graphicx}

\else

 \usepackage[dvips]{graphicx}

\fi

\usepackage{amsmath,amssymb,amsfonts,latexsym}
\usepackage{upgreek}
\newtheorem{lemma}{Lemma}
\newtheorem{remark}{Remark}

\newenvironment{proof}{{\hspace{0.3cm}\it Proof:}\hspace{0.15cm}}{\par}
\allowdisplaybreaks[4]

\usepackage{algorithm,float}
\usepackage{algorithmic}
\usepackage{lipsum}


\def\mathbi#1{\textbf{\em #1}}
\usepackage{bbding}

\usepackage{array}

\usepackage{fixltx2e}

\usepackage{comment}

\usepackage{bm}

\usepackage{stfloats}

\begin{document}

\title{Spatial Coded Modulation}

\author{Junshan~Luo,~Fanggang~Wang,~and~Shilian~Wang
\thanks{J. Luo and S. Wang are with the College of Electronic Science, National University of Defense Technology, Changsha 410073,
China (e-mail: ljsnudt@foxmail.com, wangsl@nudt.edu.cn)}
\thanks{F. Wang is with the State Key Laboratory of Rail Traffic Control and Safety, Beijing Jiaotong University, Beijing 100044, China (e-mail: wangfg@bjtu.edu.cn)}
}

\maketitle

\begin{abstract}
  In spatial modulation systems, information bits are retrieved by detecting the active antennas and the modulated symbols. The reliability of the active antenna detection is of vital importance since the modulated symbols tend to be correctly demodulated when the active antennas are accurately identified. In this paper, we propose a \textit{spatial coded modulation} (SCM) scheme, which improves the accuracy of the active antenna detection by coding over the transmit antennas. Specifically, the antenna activation pattern in the SCM corresponds to a codeword in a properly designed codebook with a larger minimum Hamming distance than its counterpart conventional spatial modulation. As the minimum Hamming distance increases, the reliability of the active antenna detection is directly enhanced, which in turn improves the demodulation of the modulated symbols and yields a better system reliability. In addition to the reliability, the proposed SCM scheme also achieves a higher capacity with the identical antenna configuration compared to the conventional spatial modulation technique. Moreover, the proposed SCM scheme strikes a balance between spectral efficiency and reliability by trading off the minimum Hamming distance with the number of available codewords. The optimal maximum likelihood detector is first formulated. Then, a low-complexity suboptimal detector is proposed to reduce the computational complexity, which has a two-step detection. Theoretical derivations of the channel capacity and the bit error rate are presented in various channel scenarios, i.e., Rayleigh, Rician, Nakagami-$m$, imperfect channel state information, and spatial correlation. Further derivation on performance bounding is also provided to reveal the insight of the benefit of increasing the minimum Hamming distance. Numerical results validate the analysis and demonstrate that the proposed SCM outperforms the conventional spatial modulation techniques in both channel capacity and system reliability.
\end{abstract}


%

\section{Introduction}
Higher spectral efficiency and better reliability are key features of the future wireless communication systems \cite{Andrews_2014}. In response to these ever-growing demands, multi-antenna techniques have been intensively investigated. In particular, communication systems employing multiple antennas at both ends can achieve better performance in terms of both reliability and efficiency.

In the existing multi-antenna schemes, the multiple-input multiple-output (MIMO) multiplexing techniques significantly improve the system capacity. However, one of the main drawbacks associated with the conventional MIMO multiplexing architecture is the requirement of multiple radio frequency (RF) chains, which leads to high power consumption and high hardware cost \cite{Renzo_M2014}. Therefore, the spatial modulation (SM) was proposed as a compelling single-RF multi-antenna scheme, which conveys information bits via both the conventional modulated symbols and the indices of the transmit antennas \cite{Mesleh_R2008}. This single-RF design offers attractive benefits over the conventional MIMO schemes in terms of  energy efficiency, hardware cost, complexity, and reliability \cite{Mesleh_R2017_1,Yang_P2015}. However, in order to utilize all the transmit antennas efficiently, the bit-to-antenna mapping requires the number of transmit antennas to be a power of two, which imposes a major limitation on the system design.  In addition, since only a single antenna transmits the signal, the SM scheme is spectrum-inefficient compared to the multiplexing techniques.



Encompassing the SM technique as a special case, the generalized spatial modulation (GSM) was proposed to break through the limitation on the number of transmit antennas and improve the spectral efficiency in the spatial domain, by activating a group of transmit antennas in each transmission and exploiting the indices of the active antenna group to deliver additional information bits \cite{Younis_A2010}. It has been proved that the GSM scheme has the advantages of the conventional SM techniques, such as the single-RF architecture \cite{Mesleh_R2017}, which has attracted intense attention in the recent years. Previous works on GSM have been mainly concerned with optimizing the error performance \cite{Cheng_P2018,Rajashekar_R2018,Xiao_L2018}, improving the system capacity \cite{Jin_S2015,Perovic_NS2018}, designing low-complexity detectors \cite{Xiao_Y2014,Liu_W2014,Rodriguez_2015,Xiao_L2017}, and implementing the practical applications \cite{Narasimhan_2015,Xiao_M2017,Liu_P2016,He_L2017,Kumar_RC2018}.


Despite the previous work, in this paper, we aim to enhance the correct detection of the active antennas. This problem is important since, on one hand, the modulated symbols tend to be correctly demodulated when the active antennas are accurately identified. Otherwise, the unreliable detection of antenna activation patterns would deteriorate the demodulation of the modulated symbols, which results in the overall system reliability degradation. This is due to retrieving the transmitted modulated symbols under incorrect estimation of the antenna activation patterns corresponds to the event where the modulated symbols need to be detected with erroneous channel estimation. On the other hand, the identification of the antenna activation patterns are even more unreliable in sparsely scattered channels due to similar channel states from transmit antenna groups to the receive antennas.


Indeed, channel coding has proved to be an effective way to enhance the system reliability and coded spatial modulation schemes have been extensively investigated in \cite{Mesleh_R2010,Basar_E2011,Basar_E2011_1,Basar_E2012,Feng_D2018,Wang_L2018}. However, although these schemes are able to enhance the system reliability, they adopt the coding scheme in a conventional way which is to encode the information bits directly. The reliability is improved at the cost of the redundancy. In contrast, the proposed SCM scheme adopts coding over the transmit antennas without sacrificing any communication efficiency.


In this paper, we propose the strategy of coding-over-antenna to directly enhance the antenna detection, which is called \textit{spatial coded modulation} (SCM). In particular, coding-over-antenna implies that the antenna activation patterns are determined by the binary error control codes, in which ``$1$'' denotes the antenna to be active and ``$0$'' represents the antenna to be idle. In the proposed SCM, the minimum Hamming distance is enlarged and the accuracy of detecting the antenna activation patterns can be enhanced. In contrast to the conventional SM and GSM schemes where a fixed number of transmit antennas is used in each transmission, the proposed SCM allows a varying number of the active transmit antennas, which is determined by the Hamming weight of each codeword. Note that the idea of varying the number of active transmit antennas was independently proposed in \cite{Osman_2015,Liu_P2017}, and in the literature of the space-shift keying (SSK) modulation, which encodes information bits onto the antenna index only. However, this SSK scheme is designed for achieving better balance between reliability and power consumption \cite{Chang_RY2012}. In contrast, the proposed SCM has a different motivation that we use coding-over-antenna to improve both the antenna detection and the sequel modulated symbol detection. Beyond \cite{Chang_RY2012}, we provide both theoretical analysis and numerical results, which show that the proposed SCM scheme also improves the system capacity.

Overall, the proposed SCM scheme mainly benefits from four aspects, i.e., reliability, capacity, flexibility and compatibility: First, it directly encodes over the transmit antennas and thus improves the accuracy of the active antenna detection, which also enhances the baseband symbol demodulation; Second, it achieves a higher capacity than the conventional SM and GSM schemes with an identical antenna configuration; Third, it is capable of achieving a flexible tradeoff between the spectral efficiency and system reliability. In particular, the minimum Hamming distance of the codebook and the achieved spectral efficiency are a tradeoff under a certain required system performance; Fourth, it is compatible with the existing SM schemes. Thus, it can be readily incorporated with these existing SM schemes to further enhance the reliability performance. Moreover, although a varying number of antennas are activated, the proposed SCM scheme could be implemented with a single RF chain since a single signaling stream is transmitted, which was verified in \cite{Mesleh_R2017}. This single-RF architecture is advantageous from the perspective of both energy efficiency and implementation cost. The contribution of this paper is summarized as follows:
\begin{itemize}
  \item A transmission scheme referred to as the spatial coded modulation is proposed to improve the reliability and the capacity without sacrificing the communication efficiency. In contrast to the conventional spatial modulation schemes, the main feature of the proposed SCM is coding-over-antenna, i.e., the antenna activation patterns are determined by the error control codes. Through coding-over-antenna, the detection probability of the active antennas could be directly enhanced and the baseband symbol demodulation is improved as well. In addition, the system capacity can also be enhanced compared to the conventional SM and GSM schemes with an identical antenna configuration.
  \item The optimal maximum likelihood (ML) detector is first formulated and then a low-complexity suboptimal detector is proposed. The suboptimal detector first determines a set of the possible indices of antenna activation patterns and then performs the ML detection to jointly estimate the antenna indices and the modulated symbols. Accordingly, a flexible tradeoff between the computational complexity and the error performance is achieved by altering the cardinality of the antenna candidate set.
  \item The capacity and the bit error rate (BER) are derived in various channel scenarios including Rayleigh, Rician, Nakagami-$m$, spatial correlation, and  imperfect channel estimation. The theoretical analysis shows that the proposed SCM scheme enables a larger minimum Hamming distance of the codebook, and thus achieves better performance in both the capacity and the reliability, which is further validated through numerical results. The gain can be further improved when the number of antennas increases. This observation exhibits the proposed SCM to be promising especially in the context of large-scale MIMO.
  \item The capacity and the reliability of the proposed SCM are numerically evaluated and compared with the conventional SM and GSM schemes in different fading channels. The results show that the improvement in reliability and capacity are more significant in the sparsely scattered channels and the highly spatially correlated channels, in which it is more demanding to detect the active antennas.
\end{itemize}

The remainder of this paper is organized as follows. The general transceiver setting of the SM schemes is presented in Section II. In Section III, we present two toy examples to illustrate the benefits of the proposed SCM scheme. In Section IV, we introduce the general framework of the proposed SCM scheme and formulate both the optimal and the suboptimal detectors. Theoretical analysis of the capacity and the reliability is provided in Section V, and the numerical results are presented in Section VI. Finally, this paper is concluded in Section VII.

\textit{Notation:} Throughout this paper, variables, vectors, and matrices are written as italic letters $x$, bold italic letters $\boldsymbol{x}$, and bold capital italic letters $\boldsymbol{X}$, respectively. Random variables, random vectors and a random matrices are respectively denoted by $\mathsf{x}$, $\boldsymbol{\mathsf{x}}$, and $\boldsymbol{\mathsf{X}}$. $|\mathcal{X}|$ is the cardinality of set $\mathcal{X}$; $\mathbb{E}$ and $\mathbb{E}_{\mathsf{x}}$ denote the expectation with respect to all the randomness in the argument and the expectation with respect to $\mathsf{x}$, respectively; $\Re\{c\}$ represents the real part of the complex number $c$. The operators $\otimes$, $|\cdot|$, $\lfloor \cdot \rfloor$, $(  \cdot  )^\dagger $, and $\left\|  \cdot  \right\|$ denote the Kronecker product, the determinant, the floor, the Hermitian, and the $\ell_2$ norm of their arguments, respectively. The operator $\textrm{vec}(\cdot)$ is to stack each column of the matrix on top of the right adjacent column. $\boldsymbol{I}_k$ denotes the $k$-by-$k$ identity matrix, and its subscript can be omitted when there is no confusion. Define $\mathcal{I}_N = \left\{1, 2, \ldots, N \right\}$ as a shorthand as the index set. The default base of the logarithm is $2$ in this paper.

\section{System Model}
In this section, we preview the system model of a general spatial modulation scheme. Consider an $M \times N$ multi-antenna communication system, where $M$ and $N$ represent the numbers of the transmit antennas and the receive antennas, respectively. From the convention of spatial modulation, we assume that one modulated symbol is transmitted and the indices of the active transmit antennas are chosen once per channel use.\footnote{The assumption is for the ease of illustration in the following. This assumption can be generalized to the case of transmitting multiple modulated symbols at each active antenna pattern, which is trivial in the investigation of the SM schemes.} Hereafter, we formulate the problem as per channel use. The received signal $\boldsymbol{\mathsf{y}} \in \mathbb{C}^{N}$ can be formulated as
\begin{equation}\label{system_overall_model}
\boldsymbol{\mathsf{y}} = \boldsymbol{\mathsf{H}}\boldsymbol{\mathsf{z}}+\boldsymbol{\mathsf{u}}
\end{equation}
where $\boldsymbol{\mathsf{H}} \in \mathbb{C}^{N \times M}$ is the channel state information (CSI); $\boldsymbol{\mathsf{z}} \in \mathbb{C}^{M}$ is the transmit signal; $\boldsymbol{\mathsf{u}} \in \mathbb{C}^{N}$ is the noise vector with independent and identically distributed (i.i.d.) samples following the circularly symmetric complex Gaussian (CSCG) distribution, denoted by $\mathcal{CN}(0,\sigma^2)$, where $\sigma^2$ is the noise variance. For a general spatial modulation setting at the transmitter side, the transmitter activates $W$ $(1 \le W < M)$ of the $M$ transmit antennas per channel use to transmit an identical data stream. Note that, although multiple antennas are active, only a single RF chain is required at the transmitter side, which was claimed in \cite{Mesleh_R2017}. An amplitude and/or phase modulation (APM) symbol of the constellation diagram $\mathcal{S}$, $\mathsf{s} \in \mathcal{S}$, is sent from each active antenna with the normalized power, i.e., $\mathbb{E}\,|\mathsf{s}|^2=1$. Let $\boldsymbol{\mathsf{a}} = [\mathsf{a}_1, \ldots, \mathsf{a}_M]^\textrm{T}$ denote the antenna activation pattern, and then the transmit signal can be written as
\begin{equation}
\boldsymbol{\mathsf{z}} = \frac{1}{\sqrt{W}}\boldsymbol{\mathsf{a}}\mathsf{s}
\end{equation}
where $\mathsf{a}_\ell \in \{0,1\}$, $\ell \in \mathcal{I}_M$; $\mathsf{a}_\ell = 1$ represents that the $\ell$th transmit antenna is active; $\mathsf{a}_\ell = 0$ represents that the $\ell$th transmit antenna is idle. Assuming the CSI is perfectly known at the receiver, the optimal ML detection jointly determines the antenna activation pattern and the APM symbol by solving
\begin{equation}\label{ML_Detection}
(\hat{\boldsymbol{a}}, \hat s) = \arg \mathop {\min }\limits_{\boldsymbol{\mathsf{a}} \in \mathcal{A}, \mathsf{s} \in \mathcal{S}} {\Big\| {\boldsymbol{\mathsf{y}} - \frac{1}{\sqrt{W}}\boldsymbol{\mathsf{H}}\boldsymbol{\mathsf{a}}\mathsf{s}} \Big\|^2}
\end{equation}
where $\mathcal{A}$ represents all candidates of $\boldsymbol{\mathsf{a}}$. From \eqref{ML_Detection}, it can be observed that the error decision of the antenna activation patterns further deteriorates the demodulation of the APM symbols. This is due to the APM symbols are demodulated with incorrect CSI when the error decision of antenna activation patterns occurs. Moreover, in addition to the optimal ML detection, the suboptimal scheme detects the antennas and the APM symbols individually.
\begin{figure}[!t]
  \centering
  \includegraphics[width=3.5in]{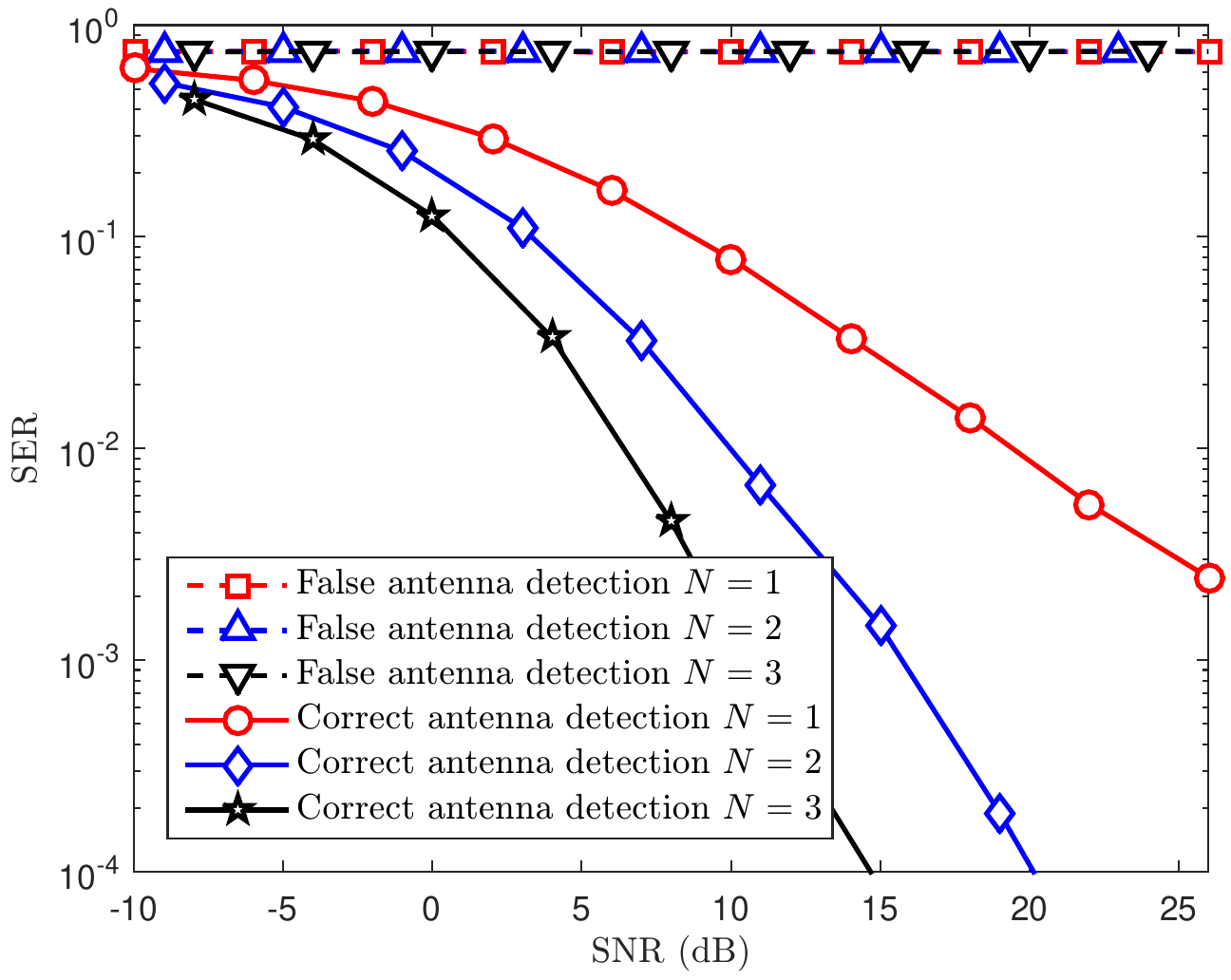}\\
  \caption{The SER of the QPSK symbols in the Rayleigh channel, where the $2 \times N$ SM system employs the ML detector. The results show that correct antenna index detection enhances the reliability of the APM symbols. In addition, the demodulation of the APM symbols cannot be improved with more receive antennas when the antenna index is erroneously detected.}
  \label{antenna_vs_APM}
\end{figure}
In this scheme, the maximum ratio combining or the matched filter is used, and the indices of the active antennas are detected first. Next, assuming that the antenna indices are correctly detected, the APM symbols are demodulated using the ML criterion. Therein, the APM symbol demodulation tends to be degraded by the error detection of the antenna indices. A quantitative illustration of this fact is shown in Fig. \ref{antenna_vs_APM}, which depicts the symbol error rate (SER) of the APM symbols in a $2 \times N$ SM scheme with both the correct and the false antenna detection. From Fig. \ref{antenna_vs_APM}, we observe that: First, correct antenna index detection significantly enhances the reliability of the APM symbols, i.e., the SER of the APM symbols with correct antenna detection outperforms its counterpart with false antenna detection in all signal-to-noise ratio (SNR) regime; Second, the demodulation of the APM symbols cannot be improved by employing more receive antennas since the incorrect CSI is used at the receiver.

\section{Toy Examples}
In this section, we use two toy examples to illustrate the principle of the proposed SCM. It is capable of correcting part of the antenna detection errors and thus outperforms the SM and the GSM schemes in reliability. Beyond that, the channel capacity is also improved. In the following, the proposed SCM is compared to the conventional SM and the GSM, respectively.
\subsection{SCM vs. SM}
We consider a $4 \times 2$ SM system with a spectral efficiency of $5$ bits/sec/Hz as the benchmark.\footnote{The reason that we choose $4$ transmit antennas is that the SM scheme generally requires the number of transmit antennas to be a power of $2$. In addition, when the transmitter has $2$ antennas, the proposed SCM reduces to the conventional SM scheme.} Herein, $2$ information bits are assigned to select one of the $4$ transmit antennas, and $3$ bits are modulated into a quadrature amplitude modulation (QAM) symbol. Specifically, the bit-to-antenna mapping of the conventional SM scheme is shown in Table I.
\begin{table}[!t]\label{SM_mapping}
\centering
\caption{Bit-to-antenna Mapping of the Conventional SM with Minimum Hamming Distance of $2$.}
\begin{tabular}{cccc}
\hline
Spatial information bits&
Index of activated Tx&
Codewords\\
\hline
$00$&
$\#1$&
$1000$\\
$01$&
$\#2$&
$0100$\\
$10$&
$\#3$&
$0010$\\
$11$&
$\#4$&
$0001$\\
\hline
\end{tabular}
\end{table}
\begin{table}[!t]\label{SCM_mapping_new}
\centering
\caption{Bit-to-antenna Mapping of the Proposed SCM with Minimum Hamming Distance of $3$.}
\begin{tabular}{cccc}
\hline
Spatial information bits&
Index of activated Tx&
Codewords\\
\hline
$0$&
$\#3,\#4$&
$0011$\\
$1$&
$\#1$&
$1000$\\
\hline
\end{tabular}
\end{table}
\begin{table}[!t]\label{SCM_mapping}
\centering
\caption{Bit-to-antenna Mapping of the Proposed SCM with Minimum Hamming Distance of $4$.}
\begin{tabular}{cccc}
\hline
Spatial information bits&
Index of activated Tx&
Codewords\\
\hline
$0$&
$\#4$&
$0001$\\
$1$&
$\#1,\#2,\#3$&
$1110$\\
\hline
\end{tabular}
\end{table}
As a fair comparison, the SCM scheme adopts the identical antenna configuration, i.e., the $4 \times 2$ multi-antenna setting, with a spectral efficiency of $5$ bits/sec/Hz as well. In contrast to the conventional SM, the antenna activation patterns of the SCM are determined by a $(4,1)$ linear block code. Specifically, we consider two $(4,1)$ codes, where the first codebook has a minimum Hamming distance of $4$ and the second codebook has a minimum Hamming distance of $3$.\footnote{Different Hamming distances are chosen to reveal the insight of the proposed scheme.} In the SCM, one bit is carried by one of the two antenna activation patterns and $4$ bits are modulated into a $16$ QAM symbol. The bit-to-antenna mapping for the proposed $4 \times 2$ SCM system is shown in Tables II and III, respectively. Comparing the codebooks in SCM and SM, we observe that the proposed SCM has a larger minimum Hamming distance than that of the conventional SM scheme. Next, the BER and the capacity\footnote{The evaluation of the capacity is provided in Section \ref{Performance_Analysis}.} of both the SCM and the SM schemes are evaluated in the Rician channel, which is shown in Fig. \ref{Compr_SM_ScodSM_Toy_Model}.\footnote{The reason that we adopt the Rician fading is that the improvement of the SCM is significant. In Rayleigh fading, the SCM still outperforms the SM with slightly smaller gain.} Numerical results demonstrate that the SCM outperforms the conventional SM in terms of both the reliability and the efficiency. Moreover, the gain is improved as the minimum Hamming distance increases.
\begin{figure}[!t]
  \centering
  \includegraphics[width=3.5in]{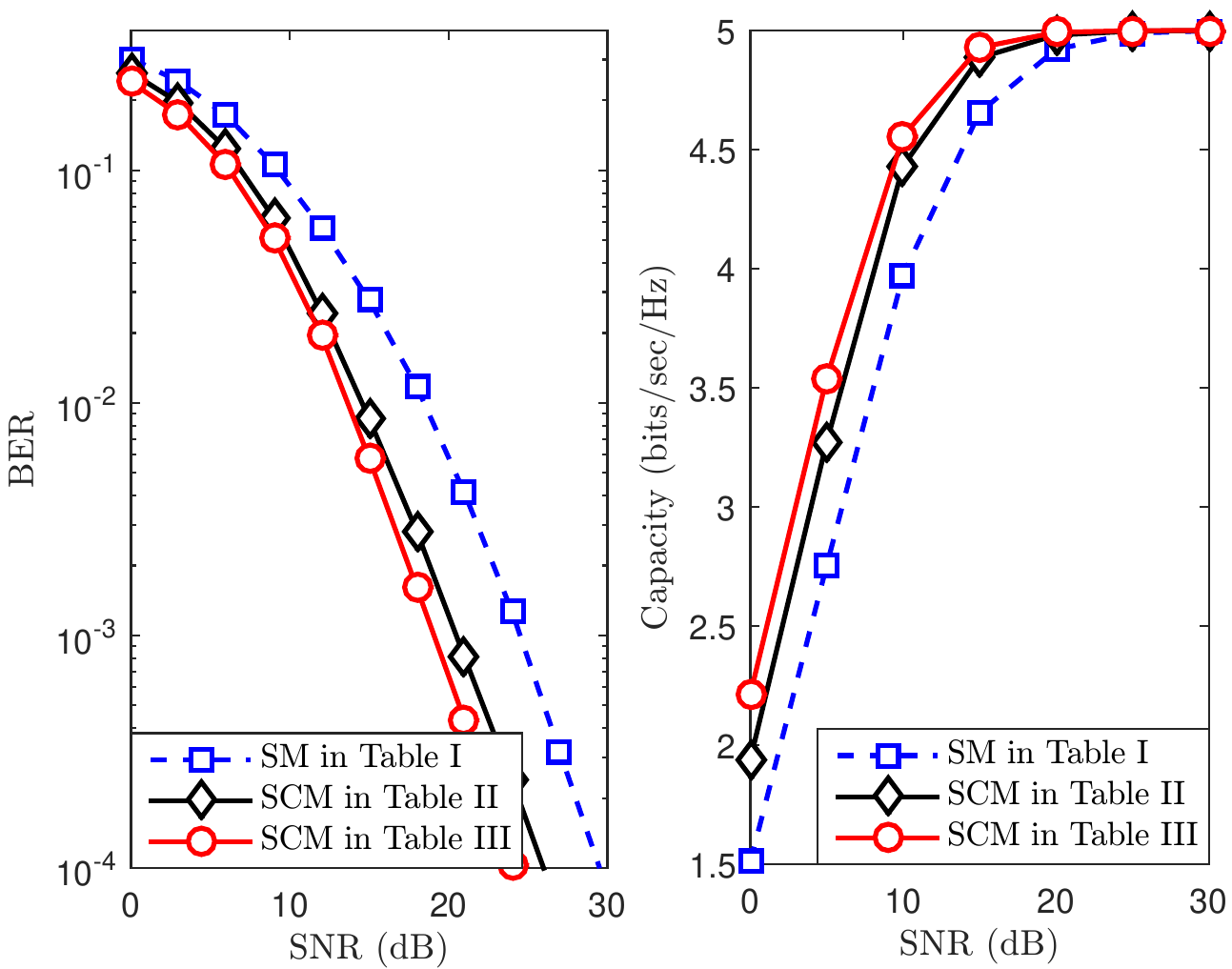}\\
  \caption{The BER and the capacity of the proposed SCM and the conventional SM in the Rician channel with a spectral efficiency of $5$ bits/sec/Hz. The transmitter has $4$ antennas and the receiver has $2$ antennas. The bit-to-antenna mapping tables are shown in Tables I, II and III, respectively. The results reveal that the SCM outperforms the conventional SM in both the reliability and the efficiency.}
  \label{Compr_SM_ScodSM_Toy_Model}
\end{figure}
\subsection{SCM vs. GSM}
Since $3$ transmit antennas are the minimum setup of the GSM, we consider a $3 \times 2$ GSM system having $2$ active antennas with a spectral efficiency of $3$ bits/sec/Hz as the benchmark. One information bit is to select one of the two antenna activation patterns and the other two bits are modulated into a quaternary phase-shift keying (QPSK) symbol. Specifically, the bit-to-antenna mapping of the GSM scheme is shown in Table IV. As a fair comparison, the SCM adopts the identical antenna configuration, i.e., the $3 \times 2$ multi-antenna setting, with a spectral efficiency of $3$ bits/sec/Hz as well. In contrast to the conventional GSM, the antenna activation patterns are determined by a $(3,1)$ code. In the SCM, one bit is mapped onto one of the two antenna activation patterns and the other two bits are modulated into a QPSK symbol. The bit-to-antenna mapping for the proposed $3 \times 2$ SCM scheme is shown in Table V. Comparing the two codebooks in Tables IV and V, we observe that the minimum Hamming distance is enlarged with the proposed SCM. Next, the BER and the capacity of the SCM and the GSM schemes are evaluated in the Rician fading channel, which is shown in Fig. \ref{Compr_GSM_ScodSM_Toy_Model}. Numerical results show that the SCM outperforms the conventional GSM in both the reliability and the efficiency.
\begin{table}[!t]
\centering
\caption{Bit-to-antenna Mapping of the Conventional GSM with Minimum Hamming Distance of 2.}
\begin{tabular}{cccc}
\hline
Spatial information bits&
Index of activated Tx&
Codewords\\
\hline
$0$&
$\#1,\#2$&
$110$\\
$1$&
$\#2,\#3$&
$011$\\
\hline
\end{tabular}
\end{table}
\begin{table}[!t]
\centering
\caption{Bit-to-antenna Mapping of the Proposed SCM with Minimum Hamming Distance of 3.}
\begin{tabular}{cccc}
\hline
Spatial information bits&
Index of activated Tx&
Codewords\\
\hline
$0$&
$\#3$&
$001$\\
$1$&
$\#1,\#2$&
$110$\\
\hline
\end{tabular}
\end{table}
\begin{figure}[!t]
  \centering
  \includegraphics[width=3.5in]{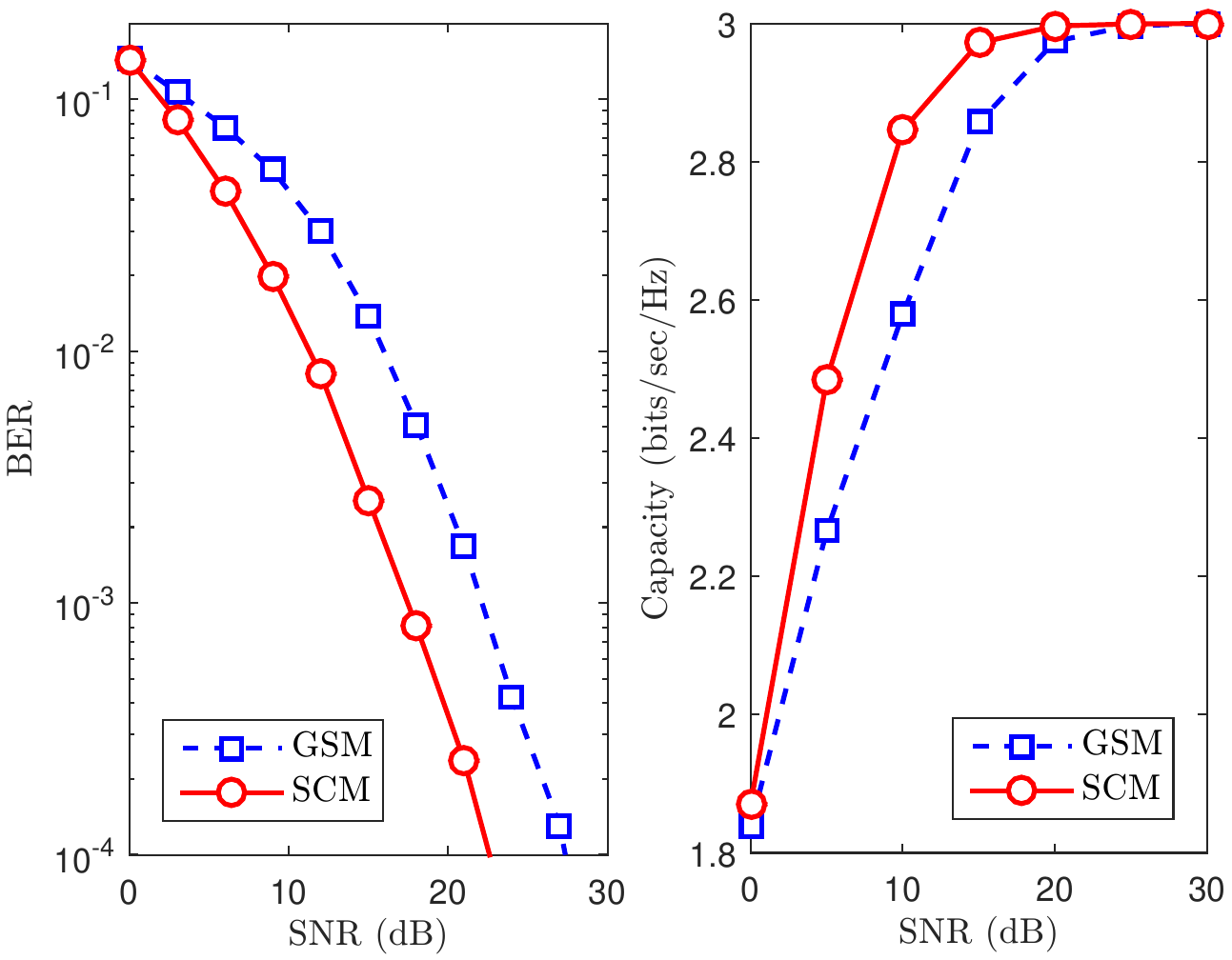}\\
  \caption{The BER and the capacity of the proposed SCM and the conventional GSM in the Rician channel with a spectral efficiency of $3$ bits/sec/Hz. The transmitter has $3$ antennas and the receiver has $2$ antennas. The bit-to-antenna mapping tables are shown in Tables IV and V, respectively. The results reveal that the SCM outperforms the conventional GSM in both the reliability and the efficiency.}
  \label{Compr_GSM_ScodSM_Toy_Model}
\end{figure}
\begin{remark}
The reliability and the efficiency of the proposed SCM can be further improved using more transmit antennas, since the minimum Hamming distance can be further enlarged. It is worth noting that, although multiple antennas are deployed and are activated simultaneously at the transmitter, the proposed SCM scheme can be implemented with only a single RF chain since each active antenna transmits the identical signal, which was claimed in \cite{Mesleh_R2017}.
\end{remark}
\section{Transceiver Design of SCM}
In this section, we introduce the design of the proposed SCM scheme, the optimal and the suboptimal detectors. The input data bits are split into two streams, i.e., the signaling data stream and the spatial data stream. The bits of the signaling data stream are modulated by an APM constellation and the bits of the spatial data stream are mapped onto the antenna activation patterns. In particular, the antenna activation patterns are designed corresponding to a $(M,k)$ error control code, where the one in a codeword indicates the corresponding antenna to be active and the zero represents the corresponding antenna to be idle. Note that the code length is the number of transmit antennas. In the following, we use the notation SCM $(M,k)$ to specify that $k$ bits are conveyed in the spatial domain. For the SCM $(M,k)$ employing the constellation diagram $\mathcal{S}$, a vector of $\log L$ bits can be delivered in each transmission, where $\log L = k + \log |\mathcal S|$. Specifically, in a bit stream, the first $k$ bits are encoded into $M$ bits, which serves as an antenna activation pattern. The remaining $\log |\mathcal S|$ information bits are modulated into an APM symbol. Let $\boldsymbol{c}$ denote a generic $(M,k)$ codeword, and the $(M,k)$ codebook $\mathcal{C}$ is with the cardinality of $2^k$. Note that the trivial all-zero codeword is excluded since at least one transmit antenna is required to be active. By properly designing the codebook $\mathcal{C}$, the minimum Hamming distance can be enlarged and the reliability of antenna detection can be improved accordingly.

Assuming the APM symbol $\mathsf{s}$ and the codeword $\boldsymbol{\mathsf{c}} \in \mathcal{C}$ are transmitted, the received signal $\boldsymbol{\mathsf{r}} \in \mathbb{C}^N$ is formulated as
\begin{align}
\boldsymbol{\mathsf{r}} &= \boldsymbol{\mathsf{H}}\boldsymbol{\mathsf{x}}+\boldsymbol{\mathsf{u}} \label{receive_signal} \\
\label{y_SCM}&= \frac{1}{\sqrt{w(c)}}\boldsymbol{\mathsf{H}}\boldsymbol{\mathsf{c}}\mathsf{s}+\boldsymbol{\mathsf{u}}
\end{align}
where $\boldsymbol{\mathsf{x}} \in \mathcal{X}$ denotes the resultant SCM symbol with the transmit power normalized, i.e., $\mathbb{E}\,\|\boldsymbol{\mathsf{x}}\|^2 = 1$; $\mathcal{X}$ represents the alphabet of the SCM symbol; $w(c)$ is the Hamming weight of the codeword $\boldsymbol{\mathsf{c}}$.
\begin{remark}
A flexible transmission rate of the spatial information is achieved by the $(M,k)$ codebook, ranging from $1$ to $M-1$ bits/sec/Hz. This design also allows a practical tradeoff between the system reliability and the efficiency. For example, the bit-to-antenna mapping can be implemented by the $(7,4)$ and $(7,3)$ Hamming codes, in which the spectral efficiencies in the spatial domain are $4$ and $3$ bits/sec/Hz, respectively, and the minimum Hamming distances are $3$ and $4$, respectively.\footnote{They will be compared to the existing GSM scheme in the simulation in Section \ref{Performance_Analysis}. }
\end{remark}
\begin{remark}
The operation of coding-over-antenna is compatible with most of the existing performance enhancing techniques, such as the conventional coded spatial modulation schemes, which apply the channel coding in the time domain to improve the reliability of the detection of the APM symbols directly.
\end{remark}
\subsection{Optimal ML Detection}
The optimal detector follows from the ML criterion since each of the SSM symbols is transmitted with equal probability. Assuming the perfect CSI at the receiver, the APM symbols and the spatial codewords are jointly detected using the ML detector, i.e.,
\begin{align}
(\hat{\boldsymbol{c}}, \hat s) &= \arg \mathop {\min }\limits_{\boldsymbol{\mathsf{c}} \in \mathcal{C}, \mathsf{s} \in \mathcal{S}} {\big\| {\boldsymbol{\mathsf{r}} - \frac{1}{\sqrt{w(c)}}\boldsymbol{\mathsf{H}}{\boldsymbol{\mathsf{c}}\mathsf{s}} } \big\|^2}\\
&= \arg \mathop {\min }\limits_{\boldsymbol{\mathsf{c}} \in \mathcal{C}, \mathsf{s} \in \mathcal{S}} {\frac{1}{\sqrt{w(c)}}\|\boldsymbol{\mathsf{H}}{\boldsymbol{\mathsf{c}}}\mathsf{s}\|^2 - 2\Re{\{\boldsymbol{\mathsf{r}}^\dagger \boldsymbol{\mathsf{H}}{\boldsymbol{\mathsf{c}}}\mathsf{s}\}}}. \label{Optimal_Detection}
\end{align}
From \eqref{Optimal_Detection}, the complexity of the optimal ML detector is $\mathcal{O}(L + N2^{k+1})$, which grows rapidly with the size of the spatial constellation, i.e., $2^k$.

\subsection{Suboptimal Detection}
Since the ML detection is computationally intensive,
in this paper, we propose a low complexity detector which has the following two steps. The first step is to collect a set of possible antenna activation patterns and the second step is to use the ML detection to jointly estimate the transmit antenna indices and the modulated symbols. The searching space is reduced by excluding part of the antenna activation patterns and thus the computational complexity is reduced. Let $\boldsymbol{\mathsf{t}}$ denote the product of the channel $\boldsymbol{\mathsf{H}}$ and the codeword $\boldsymbol{\mathsf{c}}$
\begin{equation}
\boldsymbol{\mathsf{t}} = \boldsymbol{\mathsf{H}}\boldsymbol{\mathsf{c}}.
\end{equation}

In the absence of noise, the received signal $\boldsymbol{\mathsf{r}}$ in \eqref{y_SCM} is equal to $\boldsymbol{\mathsf{t}}$. Therefore, a small bias or angle between $\boldsymbol{\mathsf{r}}$ and $\boldsymbol{\mathsf{t}}$ indicates that the codeword $\boldsymbol{\mathsf{c}}$ tends to be transmitted. Specifically, let $\alpha$ denote the angle between $\boldsymbol{\mathsf{t}}$ and $\boldsymbol{\mathsf{r}}$, and we have
\begin{equation}
\alpha = \arccos \frac{|\langle\boldsymbol{\mathsf{t}}, \boldsymbol{\mathsf{r}}\rangle|}{\|\boldsymbol{\mathsf{t}}\| \cdot \|\boldsymbol{\mathsf{r}}\|}.
\end{equation}
Then a candidate set of $\boldsymbol{\mathsf{t}}$ can be determined, which has the smaller angles than others. Let $\mathcal{T}$ denote the candidate set. For each possible $\boldsymbol{\mathsf{t}}$ in the candidate set $\mathcal{T}$, the antenna indices and the APM symbol are jointly detected as
\begin{equation}
(\hat{\boldsymbol{t}}, \hat s) = \arg \min\limits_{\boldsymbol{\mathsf{t}} \in \mathcal{T}, \mathsf{s} \in \mathcal{S}} {\Big\| {\boldsymbol{\mathsf{r}} - \frac{1}{\sqrt{w(c)}}\boldsymbol{\mathsf{t}}\mathsf{s} } \Big\|^2}.
\end{equation}

\section{Capacity and BER Analysis}\label{Performance_Analysis}
In this section, the capacity and the reliability of the proposed SCM are analyzed in various channel scenarios, including Rayleigh, Rician, Nakagami-$m$, the spatial correlation, and the channel uncertainties. It is shown that these two performance metrics are both improved as the minimum Hamming distance increases.
\subsection{Capacity Analysis}
The SCM symbols are with finite input of the wireless channel since they are constituted by both the antenna activation pattern indices and the APM modulated symbols. Regarding the output of the channel, the received signals are continuous due to the fading and the additive Gaussian noise. Consequently, the discrete-input continuous-output capacity is derived for the SCM scheme.\footnote{Similar analysis was conducted for the generalized precoding aided spatial modulation in \cite{Zhang_R2015}. In this paper, we follow the convention of the derivation and further provide a tractable capacity lower bound.} In the following, we first present the exact expression of the capacity. Then, we provide a lower bound to exhibit the effect of the minimum Hamming distance in the capacity.
\subsubsection{Derivation of Capacity}
The capacity of the SCM is evaluated by maximizing the mutual information between the input and output of the channel. Specifically, we provide Lemma $1$ as follows:
\begin{lemma}
The ergodic capacity of the proposed SCM is given by
\begin{equation}\label{exact_capacity}
C = \log L - \frac{1}{L}\sum\limits_{i = 1}^L {\mathbb{E}_{\boldsymbol{\mathsf{H}},\boldsymbol{\mathsf{u}}}\log\sum\limits_{j=1}^Le^{ \frac{1}{\sigma^2} (\|\boldsymbol{\mathsf{u}}\|^2 - \|\boldsymbol{\mathsf{u}}+\boldsymbol{\mathsf{H}}(\boldsymbol{\mathsf{x}}_j - \boldsymbol{\mathsf{x}}_i)\|^2)}}.
\end{equation}
\end{lemma}
\begin{proof}
The proof follows from the following derivation
\begin{align}
C &= \mathbb{E}_{\boldsymbol{\mathsf{H}}}\max\limits_{\{p(\boldsymbol{\mathsf{x}}_i)\}_{1}^L} I(\boldsymbol{\mathsf{x}};\boldsymbol{\mathsf{r}}|\boldsymbol{\mathsf{H}})\\
&= \frac{1}{L}\sum\limits_{i = 1}^L \mathbb{E}_{\boldsymbol{\mathsf{H}}} \int_{ \boldsymbol{\mathsf{r}} } {p(\boldsymbol{\mathsf{r}}|\boldsymbol{\mathsf{x}}_i, \boldsymbol{\mathsf{H}})\log {\frac{{p\left( \boldsymbol{\mathsf{r}} | \boldsymbol{\mathsf{x}}_i, \boldsymbol{\mathsf{H}} \right)}}{{\sum\nolimits_{j = 1}^L {p(\boldsymbol{\mathsf{r}},\boldsymbol{\mathsf{x}}_j|\boldsymbol{\mathsf{H}})} }}} }d\boldsymbol{\mathsf{r}} \label{DCMC_Capacity} \\
&= \frac{1}{L}\sum\limits_{i = 1}^L \mathbb{E}_{\boldsymbol{\mathsf{H}}} \int_{ \boldsymbol{\mathsf{r}} } {p(\boldsymbol{\mathsf{r}}|\boldsymbol{\mathsf{x}}_i, \boldsymbol{\mathsf{H}})\log {\frac{p\left( \boldsymbol{\mathsf{r}} | \boldsymbol{\mathsf{x}}_i, \boldsymbol{\mathsf{H}} \right)}{\sum\nolimits_{j = 1}^L {p(\boldsymbol{\mathsf{r}}|\boldsymbol{\mathsf{x}}_j, \boldsymbol{\mathsf{H}})p(\boldsymbol{\mathsf{x}}_j|\boldsymbol{\mathsf{H}})}}} }d\boldsymbol{\mathsf{r}} \label{complex_capacity_1}\\
&= -\frac{1}{L}\sum\limits_{i = 1}^L \mathbb{E}_{\boldsymbol{\mathsf{H}}} \int_{ \boldsymbol{\mathsf{r}} } {p(\boldsymbol{\mathsf{r}}|\boldsymbol{\mathsf{x}}_i, \boldsymbol{\mathsf{H}})\log \frac{1}{L}\sum\limits_{j = 1}^L{\frac{ {p(\boldsymbol{\mathsf{r}}|\boldsymbol{\mathsf{x}}_j, \boldsymbol{\mathsf{H}})}}{p\left( \boldsymbol{\mathsf{r}} | \boldsymbol{\mathsf{x}}_i, \boldsymbol{\mathsf{H}} \right)}} }d\boldsymbol{\mathsf{r}} \label{capacity_interchange}\\
&= \log L - \frac{1}{L}\sum\limits_{i = 1}^L \mathbb{E}_{\boldsymbol{\mathsf{H}}} {\int_{ \boldsymbol{\mathsf{r}} } {p(\boldsymbol{\mathsf{r}}|\boldsymbol{\mathsf{x}}_i,\boldsymbol{\mathsf{H}})}\log\sum\limits_{j=1}^Le^{ \frac{1}{\sigma^2} (\|\boldsymbol{\mathsf{u}}\|^2 - \|\boldsymbol{\mathsf{u}}+\boldsymbol{\mathsf{H}}(\boldsymbol{\mathsf{x}}_j - \boldsymbol{\mathsf{x}}_i)\|^2)} d\boldsymbol{\mathsf{r}}} \label{capacity_3} \\
&= \log L - \frac{1}{L}\sum\limits_{i = 1}^L \mathbb{E}_{\boldsymbol{\mathsf{H}}} {\int_{ \boldsymbol{\mathsf{u}} } {p(\boldsymbol{\mathsf{u}})}\log\sum\limits_{j=1}^Le^{ \frac{1}{\sigma^2} (\|\boldsymbol{\mathsf{u}}\|^2 - \|\boldsymbol{\mathsf{u}}+\boldsymbol{\mathsf{H}}(\boldsymbol{\mathsf{x}}_j - \boldsymbol{\mathsf{x}}_i)\|^2)} d\boldsymbol{\mathsf{u}}} \label{capacity_4} \\
&= \log L - \frac{1}{L}\sum\limits_{i = 1}^L {\mathbb{E}_{\boldsymbol{\mathsf{H}},\boldsymbol{\mathsf{u}}}\log\sum\limits_{j=1}^Le^{ \frac{1}{\sigma^2} (\|\boldsymbol{\mathsf{u}}\|^2 - \|\boldsymbol{\mathsf{u}}+\boldsymbol{\mathsf{H}}(\boldsymbol{\mathsf{x}}_j - \boldsymbol{\mathsf{x}}_i)\|^2)}} \label{capacity_5}
\end{align}
where $p(\cdot|\cdot)$ is the conditional probability density function; $I(\boldsymbol{\mathsf{x}};\boldsymbol{\mathsf{r}}|\boldsymbol{\mathsf{H}})$ is the mutual information between the discrete transmit signal $\boldsymbol{\mathsf{x}}$ and the continuous received signal $\boldsymbol{\mathsf{r}}$ for a given channel. The equation \eqref{DCMC_Capacity} follows from the fact that the maximization is achieved when each symbol is transmitted with equal probability, i.e., $p(\boldsymbol{\mathsf{x}}_i) = \frac{1}{L}$, $i \in \mathcal{I}_L$, which was validated in \cite{Zhang_R2015}. The equation \eqref{complex_capacity_1} is obtained by the chain rule, and \eqref{capacity_interchange} follows from the fact that $p(\boldsymbol{\mathsf{x}}_j|\boldsymbol{\mathsf{H}}) = p(\boldsymbol{\mathsf{x}}_j) = \frac{1}{L}, j \in \mathcal{I}_L$. The equation \eqref{capacity_3} follows from
\begin{align}
\frac{p(\boldsymbol{\mathsf{r}}|\boldsymbol{\mathsf{x}}_j, \boldsymbol{\mathsf{H}})}{p(\boldsymbol{\mathsf{r}}|\boldsymbol{\mathsf{x}}_i, \boldsymbol{\mathsf{H}})} &= e^{ \frac{1}{\sigma^2} (\|\boldsymbol{\mathsf{u}}\|^2 - \|\boldsymbol{\mathsf{u}}+\boldsymbol{\mathsf{H}}(\boldsymbol{\mathsf{x}}_j - \boldsymbol{\mathsf{x}}_i)\|^2)} \label{function_1}
\end{align}
where $p({\boldsymbol{\mathsf{r}}|\boldsymbol{\mathsf{x}}}, \boldsymbol{\mathsf{H}})$ is the conditional probability of receiving $\boldsymbol{\mathsf{r}}$ given the SCM symbol $\boldsymbol{\mathsf{x}}$ and the channel $\boldsymbol{\mathsf{H}}$. Explicitly, $p({\boldsymbol{\mathsf{r}}|\boldsymbol{\mathsf{x}}}, \boldsymbol{\mathsf{H}})$ is Gaussian and is expressed as
\begin{align}
p({\boldsymbol{\mathsf{r}}|\boldsymbol{\mathsf{x}}}, \boldsymbol{\mathsf{H}}) &= \frac{1}{\pi^N \sigma^{2N}}e^{-\frac{1}{\sigma^2}\|\boldsymbol{\mathsf{r}}-{\boldsymbol{\mathsf{H}}}\boldsymbol{\mathsf{x}}\|^2} \\
&= \frac{1}{\pi^N \sigma^{2N}}e^{-\frac{1}{\sigma^2}\|\boldsymbol{\mathsf{u}}\|^2}. \label{condition_prob}
\end{align}
The equation \eqref{capacity_4} follows \eqref{condition_prob}, and \eqref{capacity_5} follows the definition of the expectation, which concludes the proof.
\end{proof}

It should be noted that the capacity in \eqref{exact_capacity} lacks tractable expression, and thus provides limited insight. The capacity is numerically evaluated with large number of iterations and the results are presented in Section \ref{Simulations}.
\subsubsection{Capacity Lower Bound}
We first provide a lemma to specify a lower bound of the capacity, which is especially tight at high SNR. Then, we illustrate the relation between the minimum Hamming distance and the capacity lower bound.
\begin{lemma}
The capacity of the proposed SCM is lower bounded by
\begin{equation}
C \ge 2\log L - \mathcal{F}(\sigma, \bm{\mu}, \bm{\varSigma}) \label{f_function}
\end{equation}
where $\bm{\mu}$ and $\bm{\varSigma}$ are the mean vector and the variance matrix of the channel matrix, respectively; The function in \eqref{f_function} is defined as
\begin{equation}\label{F_function}
\mathcal{F}(\sigma, \bm{\mu}, \bm{\varSigma}) \triangleq \log \sum\limits_{i=1}^L \sum\limits_{j=1}^L \frac{1}{|\boldsymbol{I}+\frac{1}{2\sigma^2}\bm{\varSigma}\boldsymbol{\mathsf{A}}_{ij}|}e^{-\frac{1}{2\sigma^2}\bm{\mu}^\dagger \boldsymbol{\mathsf{A}}_{ij}\big(\boldsymbol{I}+\frac{1}{2\sigma^2}\bm{\varSigma}\boldsymbol{\mathsf{A}}_{ij}\big)^{-1}\bm{\mu}}
\end{equation}
where $\boldsymbol{\mathsf{A}}_{ij} = \boldsymbol{I}_{N} \otimes \boldsymbol{\updelta}_{ij}\boldsymbol{\updelta}_{ij}^\dagger$ and we have the shorthand $\boldsymbol{\updelta}_{ij} = \boldsymbol{\mathsf{x}}_i - \boldsymbol{\mathsf{x}}_j$.
\end{lemma}
\newpage
\begin{proof}
The proof follows from the following derivation
\begin{align}
C &\ge -\mathbb{E}_{\boldsymbol{\mathsf{H}}}\log \mathbb{E}_{\boldsymbol{\mathsf{r}}}\;p(\boldsymbol{\mathsf{r}}|\boldsymbol{\mathsf{H}}) - \mathbb{E}_{\boldsymbol{\mathsf{H}}}\;h(\boldsymbol{\mathsf{r}}|\boldsymbol{\mathsf{x}},\boldsymbol{\mathsf{H}}) \label{cap_ineq_1} \\
&= -\mathbb{E}_{\boldsymbol{\mathsf{H}}}\log \mathbb{E}_{\boldsymbol{\mathsf{r}}}\;p(\boldsymbol{\mathsf{r}}|\boldsymbol{\mathsf{H}}) - N\log\pi e \sigma^2 \label{cap_ineq_2} \\
&= -\mathbb{E}_{\boldsymbol{\mathsf{H}}}\log \int_{\boldsymbol{\mathsf{r}}} \bigg(\frac{1}{L\pi^{N} \sigma^{2N}}\sum\limits_{i = 1}^L e^{-\frac{1}{\sigma^2}\|\boldsymbol{\mathsf{r}}-{\boldsymbol{\mathsf{H}}}\boldsymbol{\mathsf{x}}_i\|^2}\bigg)^2 d\boldsymbol{\mathsf{r}} - N\log\pi e \sigma^2 \label{cap_eq_3} \\
&= -\mathbb{E}_{\boldsymbol{\mathsf{H}}}\log \sum\limits_{i = 1}^L \sum\limits_{j = 1}^L \int_{\boldsymbol{\mathsf{r}}} e^{-\frac{1}{\sigma^2}\|\boldsymbol{\mathsf{r}}-{\boldsymbol{\mathsf{H}}}\boldsymbol{\mathsf{x}}_i\|^2} e^{-\frac{1}{\sigma^2}\|\boldsymbol{\mathsf{r}}-{\boldsymbol{\mathsf{H}}}\boldsymbol{\mathsf{x}}_j\|^2} d\boldsymbol{\mathsf{r}} + 2\log L\pi^{N} \sigma^{2N} - N\log\pi e \sigma^2 \label{entropy_2} \\
&= - \mathbb{E}_{\boldsymbol{\mathsf{H}}}\log \sum\limits_{i=1}^L \sum\limits_{j=1}^L e^{-\frac{1}{2\sigma^2}\|\boldsymbol{\mathsf{H}}(\boldsymbol{\mathsf{x}}_i - \boldsymbol{\mathsf{x}}_j)\|^2} + 2\log L + N\log\pi e \sigma^2 - N\log\pi e \sigma^2 \label{cap_eq_4}\\
&= 2\log L - \mathbb{E}_{\boldsymbol{\mathsf{H}}}\log \sum\limits_{i=1}^L \sum\limits_{j=1}^L e^{-\frac{1}{2\sigma^2}\|\boldsymbol{\mathsf{H}}(\boldsymbol{\mathsf{x}}_i - \boldsymbol{\mathsf{x}}_j)\|^2} \label{capacity_2} \\
&= 2\log L - \mathcal{F}(\sigma, \bm{\mu}, \bm{\varSigma}). \label{f_function_alt}
\end{align}
The inequality \eqref{cap_ineq_1} follows from rewriting the mutual information $I(\boldsymbol{\mathsf{x}};\boldsymbol{\mathsf{r}}|\boldsymbol{\mathsf{H}})$ as
\begin{equation}\label{mutual_infor}
I(\boldsymbol{\mathsf{x}};\boldsymbol{\mathsf{r}}|\boldsymbol{\mathsf{H}}) = h(\boldsymbol{\mathsf{r}}|\boldsymbol{\mathsf{H}}) - h(\boldsymbol{\mathsf{r}}|\boldsymbol{\mathsf{x}},\boldsymbol{\mathsf{H}})
\end{equation}
where $h(\cdot|\cdot)$ is the conditional entropy function. Using the Jensen's inequality yields a lower bound of $h(\boldsymbol{\mathsf{r}}|\boldsymbol{\mathsf{H}})$ as
\begin{equation}
h(\boldsymbol{\mathsf{r}}|\boldsymbol{\mathsf{H}}) \ge -\log \mathbb{E}_{\boldsymbol{\mathsf{r}}}\;p(\boldsymbol{\mathsf{r}}|\boldsymbol{\mathsf{H}}).
\end{equation}
The equation \eqref{cap_ineq_2} follows from
\begin{align}\label{entropy_1}
h(\boldsymbol{\mathsf{r}}|\boldsymbol{\mathsf{x}},\boldsymbol{\mathsf{H}}) &= -\mathbb{E}\log p({\boldsymbol{\mathsf{r}}|\boldsymbol{\mathsf{x}}}, \boldsymbol{\mathsf{H}}) \\
&= N\log\pi \sigma^2 + \frac{\log e}{\sigma^2}\mathbb{E}\ \|\boldsymbol{\mathsf{r}}-{\boldsymbol{\mathsf{H}}}\boldsymbol{\mathsf{x}}\|^2 \label{entropy_6} \\
&= N\log\pi e \sigma^2 \label{entropy_7}
\end{align}
where \eqref{entropy_6} is obtained by plugging \eqref{condition_prob} into \eqref{entropy_1}, and \eqref{entropy_7} follows from the fact that $\boldsymbol{\mathsf{r}}-{\boldsymbol{\mathsf{H}}}\boldsymbol{\mathsf{x}} = \boldsymbol{\mathsf{u}}$. Next, \eqref{cap_eq_3} follows from
\begin{align}
p(\boldsymbol{\mathsf{r}}|\boldsymbol{\mathsf{H}}) &= \frac{1}{L\pi^{N} \sigma^{2N}}\sum\limits_{i = 1}^L e^{-\frac{1}{\sigma^2}\|\boldsymbol{\mathsf{r}}-{\boldsymbol{\mathsf{H}}}\boldsymbol{\mathsf{x}}_i\|^2} \\
\mathbb{E}_{\boldsymbol{\mathsf{r}}}\;p(\boldsymbol{\mathsf{r}}|\boldsymbol{\mathsf{H}}) &= \int_{\boldsymbol{\mathsf{r}}} \bigg(\frac{1}{L\pi^{N} \sigma^{2N}}\sum\limits_{i = 1}^L e^{-\frac{1}{\sigma^2}\|\boldsymbol{\mathsf{r}}-{\boldsymbol{\mathsf{H}}}\boldsymbol{\mathsf{x}}_i\|^2}\bigg)^2 d\boldsymbol{\mathsf{r}}.
\end{align}
The equation \eqref{entropy_2} is obtained by changing the order of summation and integration in \eqref{cap_eq_3}. The derivation of \eqref{cap_eq_4} follows from
\begin{align}
\int_{\boldsymbol{\mathsf{r}}} e^{-\frac{1}{\sigma^2}\|\boldsymbol{\mathsf{r}}-{\boldsymbol{\mathsf{H}}}\boldsymbol{\mathsf{x}}_i\|^2} e^{-\frac{1}{\sigma^2}\|\boldsymbol{\mathsf{r}}-{\boldsymbol{\mathsf{H}}}\boldsymbol{\mathsf{x}}_j\|^2} d\boldsymbol{\mathsf{r}} &= e^{-\frac{1}{\sigma^2}(\|\boldsymbol{\mathsf{H}}\boldsymbol{\mathsf{x}}_i\|^2 + \|\boldsymbol{\mathsf{H}}\boldsymbol{\mathsf{x}}_j\|^2)} \int_{\boldsymbol{\mathsf{r}}} e^{-\frac{1}{\sigma^2}(2\|\boldsymbol{\mathsf{r}}\|^2 - 2\Re\{\boldsymbol{\mathsf{r}}^\dagger(\boldsymbol{\mathsf{H}}\boldsymbol{\mathsf{x}}_i + \boldsymbol{\mathsf{H}}\boldsymbol{\mathsf{x}}_j)\})} d\boldsymbol{\mathsf{r}} \label{entropy_3} \\
&= e^{-\frac{1}{\sigma^2}(\|\boldsymbol{\mathsf{H}}\boldsymbol{\mathsf{x}}_i\|^2 + \|\boldsymbol{\mathsf{H}}\boldsymbol{\mathsf{x}}_j\|^2)}\prod\limits_{n = 1}^{N} \int_{\mathsf{r}_n} e^{-\frac{1}{\sigma^2}(2\|\mathsf{r}_n\|^2 - 2\Re(\mathsf{r}_n^\dagger[\boldsymbol{\mathsf{H}}\boldsymbol{\mathsf{x}}_i + \boldsymbol{\mathsf{H}}\boldsymbol{\mathsf{x}}_j]_n))}d\mathsf{r}_n \label{sigma} \\
&= \pi^{N}\sigma^{2N}e^{-\frac{1}{2\sigma^2}\|\boldsymbol{\mathsf{H}}(\boldsymbol{\mathsf{x}}_i - \boldsymbol{\mathsf{x}}_j)\|^2-N} \label{entropy_4}
\end{align}
where we have \eqref{entropy_3} by reorganizing the terms irrelevant to $\boldsymbol{\mathsf{r}}$; \eqref{sigma} follows from the fact that the elements in $\boldsymbol{\mathsf{r}} = [\mathsf{r}_1, \cdots, \mathsf{r}_{N}]^\textrm{T}$ are i.i.d. samples; \eqref{entropy_4} is attained by evaluating the integration. Finally, \eqref{capacity_2} and \eqref{f_function_alt} are obtained by removing the term $N\log\pi e\sigma^2$ and using the moment generating function (MGF) of the variate $\lambda = \| \boldsymbol{\mathsf{H}}\boldsymbol{\updelta}_{ij} \|^2$, i.e.,
\begin{equation}\label{good}
\mathbb{E}_{\boldsymbol{\mathsf{H}}}\log \sum\limits_{i=1}^L \sum\limits_{j=1}^L e^{-\frac{1}{2\sigma^2}\|\boldsymbol{\mathsf{H}}\boldsymbol{\updelta}_{ij}\|^2} = \log \sum\limits_{i=1}^L \sum\limits_{j=1}^L \mathcal{M}_\lambda\Big(-\frac{1}{2\sigma^2}\Big)
\end{equation}
where the MGF of a variate $\vartheta$ is defined as \cite{Simon_MK2005}
\begin{equation}\label{MGF_2}
\mathcal{M}_\vartheta(t) \triangleq \int_{0}^{\infty}p(\vartheta)e^{t\vartheta}d\vartheta.
\end{equation}
The Frobenius norm of $\| \boldsymbol{\mathsf{H}}\boldsymbol{\updelta}_{ij} \|^2$ can be further expanded as
\begin{align}
\| \boldsymbol{\mathsf{H}}\boldsymbol{\updelta}_{ij} \|^2 &= \textrm{vec}(\boldsymbol{\mathsf{H}})^\dagger(\boldsymbol{I}_{N} \otimes \boldsymbol{\updelta}_{ij}\boldsymbol{\updelta}_{ij}^\dagger)\textrm{vec}(\boldsymbol{\mathsf{H}})\\
&\triangleq \boldsymbol{\upomega}^\dagger\boldsymbol{\mathsf{A}}_{ij}\boldsymbol{\upomega}
\end{align}
where we have the vector $\boldsymbol{\upomega} = \textrm{vec}(\boldsymbol{\mathsf{H}})$ and the Hermitian matrix $\boldsymbol{\mathsf{A}}_{ij} = \boldsymbol{I}_{N} \otimes \boldsymbol{\updelta}_{ij}\boldsymbol{\updelta}_{ij}^\dagger$. The MGF of $\boldsymbol{\upomega}^\dagger\boldsymbol{\mathsf{A}}_{ij}\boldsymbol{\upomega}$ can be expressed as
\begin{equation}\label{MGF_1}
\mathcal{M}_\lambda(t) = \frac{1}{|\boldsymbol{I}-t\bm{\varSigma}\boldsymbol{\mathsf{A}}_{ij}|}e^{t\bm{\mu}^\dagger \boldsymbol{\mathsf{A}}_{ij}\big(\boldsymbol{I}-t\bm{\varSigma}\boldsymbol{\mathsf{A}}_{ij}\big)^{-1}\bm{\mu}}.
\end{equation}
Plugging \eqref{MGF_1} into \eqref{good}, we have a general lower bound of the capacity as shown in \eqref{f_function}, which concludes the proof.
\end{proof}

Based on Lemma $2$, the following remark illustrates the relation between the minimum Hamming distance and the capacity lower bound.
\begin{remark}
The capacity lower bound of the proposed SCM improves as the minimum Hamming distance increases.
\end{remark}

The remark can be verified by manipulating the lower bound. We consider the fading channel where the capacity of the SCM is higher than that of other fading channels, i.e., the Rayleigh fading channel. Specifically, the elements of the Rayleigh fading channel are i.i.d. samples following the CSCG distribution. Hence $\bm{\mu} = \boldsymbol{0}$ and $\bm{\varSigma} = \boldsymbol{I}$. From Lemma $2$, we obtain a simple lower bound of the system capacity as
\begin{align}
C &\ge 2\log L - \log \sum\limits_{i=1}^L \sum\limits_{j=1}^L \bigg( 1 + \frac{\|\boldsymbol{\updelta}_{ij}\|^2}{2\sigma^2} \bigg)^{-N}\\
&\ge \log L - \log\bigg(1 + \frac{L-1}{2}\bigg( 1 + \frac{1}{2\sigma^2}\min\limits_{i\neq j}\|\boldsymbol{\updelta}_{ij}\|^2 \bigg)^{-N}\bigg) \label{capacity_lower_1}
\end{align}
where $\min\limits_{i\neq j}\|\boldsymbol{\updelta}_{ij}\|^2 = \min\limits_{i\neq j} \|\boldsymbol{\mathsf{x}}_i - \boldsymbol{\mathsf{x}}_j\|^2$; \eqref{capacity_lower_1} follows from replacing $\|\boldsymbol{\updelta}_{ij}\|^2$ with $\min\limits_{i\neq j}\|\boldsymbol{\updelta}_{ij}\|^2$ and thus is upper bounded by
\begin{align}
\sum\limits_{i=1}^L \sum\limits_{j=1}^L \bigg( 1 + \frac{\|\boldsymbol{\updelta}_{ij}\|^2}{2\sigma^2} \bigg)^{-N} &= L + \sum\limits_{i=1}^L \sum\limits_{j=1, j \neq i}^L \bigg( 1 + \frac{\|\boldsymbol{\updelta}_{ij}\|^2}{2\sigma^2} \bigg)^{-N}\\
&\le L + \frac{1}{2}L(L-1)\bigg( 1 + \frac{1}{2\sigma^2}\min\limits_{i\neq j}\|\boldsymbol{\updelta}_{ij}\|^2 \bigg)^{-N}.
\end{align}
From \eqref{capacity_lower_1}, it shows that the system capacity can be improved by increasing $\min\limits_{i\neq j}\|\boldsymbol{\updelta}_{ij}\|^2$, which corresponds to increasing the minimum Hamming distance in the codebook.

\subsection{BER Analysis}
As previously discussed, the detection of the antenna indices is enhanced by coding-over-antenna, which also improves the demodulation of the APM symbols. In the following, we first present the BER upper bound of the proposed SCM scheme. Then, based on the upper bound, we reveal the effect of the minimum Hamming distance in the system reliability. Regarding the BER upper bound, we provide Lemma $3$ as follows:
\begin{lemma}
The BER of the proposed SCM is upper bounded by
\begin{equation}
{P_\mathrm{e}} \le \frac{1}{{\pi L\log L}} \mathcal{G}(\sigma, \bm{\mu}, \bm{\varSigma}) \label{ABER_2}
\end{equation}
where the function in \eqref{ABER_2} is defined as
\begin{equation}\label{G_function}
\mathcal{G}(\sigma, \bm{\mu}, \bm{\varSigma}) \triangleq \sum\limits_{i=1}^L {\sum\limits_{j=1}^L d({\boldsymbol{\mathsf{x}}_i},{\boldsymbol{\mathsf{x}}_j}) \int_0^{\frac{\pi}{2}}\frac{\exp(-\frac{1}{4\sigma^2\sin^2\theta}\bm{\mu}^\dagger \boldsymbol{\mathsf{A}}_{ij}(\boldsymbol{I}+\frac{1}{4\sigma^2\sin^2\theta}\bm{\varSigma}\boldsymbol{\mathsf{A}}_{ij})^{-1}\bm{\mu})}{|\boldsymbol{I}+\frac{1}{4\sigma^2\sin^2\theta}\bm{\varSigma}\boldsymbol{\mathsf{A}}_{ij}|}d\theta }
\end{equation}
and $d({\boldsymbol{\mathsf{x}}_i},{\boldsymbol{\mathsf{x}}_j})$ denotes the Hamming distance between $\boldsymbol{\mathsf{x}}_i$ and $\boldsymbol{\mathsf{x}}_j$.
\end{lemma}
\begin{proof}
The proof follows from the following derivation
\begin{align}
{P_\mathrm{e}} &\le \frac{1}{{L\log L}}\sum\limits_{i=1}^L {\sum\limits_{j=1}^L d({\boldsymbol{\mathsf{x}}_i},{\boldsymbol{\mathsf{x}}_j})\mathbb{E}_{\boldsymbol{\mathsf{H}}}\;\mathrm{Pr}( {\boldsymbol{\mathsf{x}}_i} \to {\boldsymbol{\mathsf{x}}_j} | \boldsymbol{\mathsf{H}} )} \label{ABER_Upper_Bound}\\
&= \frac{1}{{L\log L}}\sum\limits_{i=1}^L \sum\limits_{j=1}^L d({\boldsymbol{\mathsf{x}}_i},{\boldsymbol{\mathsf{x}}_j})\mathbb{E}_{\boldsymbol{\mathsf{H}}}\; \mathrm{Pr}(\| \boldsymbol{\mathsf{r}}- \boldsymbol{\mathsf{H}}\boldsymbol{\mathsf{x}}_i \|^2 > \| \boldsymbol{\mathsf{r}}- \boldsymbol{\mathsf{H}}\boldsymbol{\mathsf{x}}_j \|^2 | \boldsymbol{\mathsf{H}}) \label{PEP_1} \\
&= \frac{1}{{L\log L}}\sum\limits_{i=1}^L \sum\limits_{j=1}^L d({\boldsymbol{\mathsf{x}}_i},{\boldsymbol{\mathsf{x}}_j})\mathbb{E}_{\boldsymbol{\mathsf{H}}}\; \mathrm{Pr}(\| \boldsymbol{\mathsf{u}} \|^2 > \| \boldsymbol{\mathsf{H}}\boldsymbol{\mathsf{x}}_i - \boldsymbol{\mathsf{H}}\boldsymbol{\mathsf{x}}_j + \boldsymbol{\mathsf{u}}\|^2 | \boldsymbol{\mathsf{H}})\\
&= \frac{1}{{L\log L}}\sum\limits_{i=1}^L \sum\limits_{j=1}^L d({\boldsymbol{\mathsf{x}}_i},{\boldsymbol{\mathsf{x}}_j})\mathbb{E}_{\boldsymbol{\mathsf{H}}}\; \mathrm{Pr}(\| \boldsymbol{\mathsf{u}} \|^2 > \| \boldsymbol{\mathsf{H}}\boldsymbol{\updelta}_{ij} \|^2 - 2\Re(\boldsymbol{\mathsf{u}}^\dagger \boldsymbol{\mathsf{H}} \boldsymbol{\updelta}_{ij}) + \| \boldsymbol{\mathsf{u}} \|^2 | \boldsymbol{\mathsf{H}})\\
&= \frac{1}{{L\log L}}\sum\limits_{i=1}^L \sum\limits_{j=1}^L d({\boldsymbol{\mathsf{x}}_i},{\boldsymbol{\mathsf{x}}_j})\mathbb{E}_{\boldsymbol{\mathsf{H}}}\; \mathrm{Pr}(2\Re(\boldsymbol{\mathsf{u}}^\dagger \boldsymbol{\mathsf{H}} \boldsymbol{\updelta}_{ij}) > \| \boldsymbol{\mathsf{H}}\boldsymbol{\updelta}_{ij} \|^2 | \boldsymbol{\mathsf{H}})\\
&= \frac{1}{{L\log L}}\sum\limits_{i=1}^L {\sum\limits_{j=1}^L d({\boldsymbol{\mathsf{x}}_i},{\boldsymbol{\mathsf{x}}_j})\mathbb{E}_{\boldsymbol{\mathsf{H}}}\; Q\bigg(\sqrt{\frac{\| \boldsymbol{\mathsf{H}}\boldsymbol{\updelta}_{ij} \|^2}{2\sigma^2}}\bigg)} \label{PEP_expand}\\
&= \frac{1}{{\pi L\log L}}\sum\limits_{i=1}^L {\sum\limits_{j=1}^L d({\boldsymbol{\mathsf{x}}_i},{\boldsymbol{\mathsf{x}}_j})\mathbb{E}_{\boldsymbol{\mathsf{H}}} \int_0^{\frac{\pi}{2}} \exp\bigg(-\frac{\| \boldsymbol{\mathsf{H}}\boldsymbol{\updelta}_{ij} \|^2}{4\sigma^2\sin^2\theta}\bigg)d\theta } \label{ABER_1} \\
&= \frac{1}{{\pi L\log L}} \mathcal{G}(\sigma, \bm{\mu}, \bm{\varSigma}) \label{ABER_4}
\end{align}
where $\mathrm{Pr}( {\boldsymbol{\mathsf{x}}_i} \to {\boldsymbol{\mathsf{x}}_j} | \boldsymbol{\mathsf{H}} )$ represents the conditional pairwise error probability (PEP) with respect to the channel $\boldsymbol{\mathsf{H}}$; $Q(\cdot)$ denotes the $Q$ function, i.e., $Q(t) = \frac{1}{\sqrt{2\pi}}\int_{t}^{\infty}e^{-\frac{x^2}{2}}dx$.

The inequality \eqref{ABER_Upper_Bound} follows from the union bound technique \cite{Jeganathan_J2009}. The equations \eqref{PEP_1}-\eqref{PEP_expand} follow from mathematical manipulations and the fact that the noise $\boldsymbol{\mathsf{u}}$ follows CSCG distribution. Next, we have \eqref{ABER_1} by using the Craig's representation of the $Q$ function as
\begin{equation}\label{conditioned_PEP}
Q\bigg(\sqrt{\frac{\| \boldsymbol{\mathsf{H}}\boldsymbol{\updelta}_{ij} \|^2}{2\sigma^2}}\bigg) = \frac{1}{\pi} \int_0^{\frac{\pi}{2}} \exp\bigg(-\frac{\| \boldsymbol{\mathsf{H}}\boldsymbol{\updelta}_{ij} \|^2}{4\sigma^2\sin^2\theta}\bigg)d\theta.
\end{equation}
The equation \eqref{ABER_4} follows from \eqref{MGF_2} and \eqref{MGF_1}, and then we conclude the proof.
\end{proof}

Based on Lemma $2$, the following remark illustrates the relation between the minimum Hamming distance and the BER upper bound.
\begin{remark}
The BER upper bound of the proposed SCM decreases as the minimum Hamming distance increases.
\end{remark}

The remark can be verified by manipulating the upper bound. We consider the fading channel in which the BER of the SCM is lower than that in other fading channels, i.e., the Rayleigh fading channel. Specifically, the elements of the Rayleigh fading channel are i.i.d. samples following the CSCG distribution. Hence $\bm{\mu} = \boldsymbol{0}$ and $\bm{\varSigma} = \boldsymbol{I}$, and the BER upper bound is derived as
\begin{align}
{P_\mathrm{e}} &\le \frac{1}{{\pi L\log L}}\sum\limits_{i=1}^L {\sum\limits_{j=1}^L d({\boldsymbol{\mathsf{x}}_i},{\boldsymbol{\mathsf{x}}_j}) \int_0^{\frac{\pi}{2}} \bigg(1 + \frac{\|\boldsymbol{\updelta}_{ij}\|^2}{8\sigma^4\sin^2\theta}\bigg)^{-N} d\theta}\\
&\le \frac{(L-1)d_{\mathrm{sum}}}{{2\pi \log L}}{ \int_0^{\frac{\pi}{2}} \bigg(1 + \frac{1}{8\sigma^4\sin^2\theta}\min\limits_{i\neq j}\|\boldsymbol{\updelta}_{ij}\|^2\bigg)^{-N} d\theta} \label{BER_upper_1}\\
&\le \frac{(L-1)d_{\mathrm{sum}}}{{2\pi \log L}}{ \bigg(1 + \frac{1}{8\sigma^4}\min\limits_{i\neq j}\|\boldsymbol{\updelta}_{ij}\|^2\bigg)^{-N} }\label{BER_1}
\end{align}
where $d_{\mathrm{sum}} = \sum\limits_{i=1}^L \sum\limits_{j=1}^L d({\boldsymbol{\mathsf{x}}_i},{\boldsymbol{\mathsf{x}}_j})$. We have \eqref{BER_upper_1} by replacing $\|\boldsymbol{\updelta}_{ij}\|^2$ with $\min\limits_{i\neq j}\|\boldsymbol{\updelta}_{ij}\|^2$ and \eqref{BER_1} follows from the fact that the integrand is monotonically increasing with $\theta \in (0, \frac{\pi}{2})$. From \eqref{BER_1}, it can be observed that the BER level is lowered by increasing $\min\limits_{i\neq j}\|\boldsymbol{\updelta}_{ij}\|^2$, which corresponds to increasing the minimum Hamming distance in the codebook.

\begin{remark}
For the conventional SM scheme, it is readily to prove that the minimum Hamming distance of its codebook is $2$ since a single antenna is active in each transmission, which means that each codeword has a single non-zero element. On the other hand, for the conventional GSM schemes with multiple active antennas, the minimum Hamming distance of its codebook is also $2$, except for the special case where the transmitter has $4$ antennas and activates $2$ in each transmission. In this specific case, the minimum Hamming distance of the GSM is $4$.
\end{remark}

\subsection{Expressions for Spatially Correlated Rician and Nakagami-$m$ Channels}
Note that both the lower bound of the system capacity and the upper bound of the BER requires the evaluation of $\bm{\mu}$ and $\bm{\varSigma}$. Thus, the exact expressions are provided in the following. We first formulate the spatially correlated channel by the Kronecker model, which is expressed as \cite{Forenza_2004}
\begin{equation}
\tilde{\boldsymbol{\mathsf{H}}} = \boldsymbol{F}^{\frac{1}{2}} \boldsymbol{\mathsf{H}} \boldsymbol{G}^{\frac{1}{2}}
\end{equation}
where $\boldsymbol{F} \in \mathbb{C}^{N \times N}$ and $\boldsymbol{G} \in \mathbb{C}^{M \times M}$ denote the receiver and transmitter spatial correlation matrices, respectively, expressed as
\begin{align}
[\boldsymbol{F}]_{l,k} &= \rho^{\left| l - k \right|}, \quad l, k \in \mathcal{I}_{N} \\
[\boldsymbol{G}]_{l,k} &= \tau^{\left| l - k \right|}, \quad l, k \in \mathcal{I}_{M}
\end{align}
with $0 \le {\rho},{\tau} \le 1$. 

\noindent\underline{\em Rician Channel:} The channel matrix $\boldsymbol{\mathsf{D}} \in \mathbb{C}^{N \times M}$ is modeled as the sum of the line-of-sight (LoS) and the non-LoS (NLoS) components, i.e.
\begin{equation}
\boldsymbol{\mathsf{D}} = \sqrt {\frac{K}{{1 + K}}} \bm{\Lambda} + \sqrt {\frac{1}{{1 + K}}} \boldsymbol{\mathsf{H}}
\end{equation}
where $\bm{\Lambda}$ models the LoS component and is composed of all ones elements; $\boldsymbol{\mathsf{H}}$ models the NLoS fading whose entries follow i.i.d. CSCG distribution; $K$ is the Rician factor, which characterizes the ratio of the LoS and the NLoS components. Then we have
\begin{align}
\label{Rician_1}&\bm{\mu}^{\mathrm{Ri}} = \sqrt{\frac{K}{K+1}}\boldsymbol{B}^{\frac{1}{2}} \times \textrm{vec}(\bm{\Lambda})\\
\label{Rician_2}&\bm{\varSigma}^{\mathrm{Ri}} = \frac{\boldsymbol{B}}{K+1}
\end{align}
where the matrix $\boldsymbol{B}$ is defined as
\begin{equation}\label{R_s}
\boldsymbol{B}\triangleq \boldsymbol{F}\otimes\boldsymbol{G}.
\end{equation}
By substituting \eqref{Rician_1}, \eqref{Rician_2} and \eqref{R_s} into \eqref{f_function}, we have the capacity lower bound in the spatially correlated Rician channel as
\begin{equation}
C \ge 2\log L - \mathcal{F}(\sigma, \bm{\mu}^{\mathrm{Ri}}, \bm{\varSigma}^{\mathrm{Ri}})
\end{equation}
By substituting \eqref{Rician_1}, \eqref{Rician_2} and \eqref{R_s} into \eqref{ABER_2}, we have the BER upper bound in the spatially correlated Rician fading channel, i.e.
\begin{equation}
P_\mathrm{e} \le \frac{1}{\pi L\log L} \mathcal{G}(\sigma, \bm{\mu}^{\mathrm{Ri}}, \bm{\varSigma}^{\mathrm{Ri}}).
\end{equation}
\noindent\underline{\em Nakagami-$m$ Channel:} We have
\begin{align}
\label{Nakagami_1}\bm{\mu}^{\mathrm{Na}} &= \frac{\Gamma(\frac{m+1}{2})e^{\imath\frac{\pi}{4}}}{\Gamma(\frac{m}{2})\sqrt{\frac{m}{2}}} \times \textrm{vec}(\bm{\Lambda})\\
\label{Nakagami_2}\bm{\varSigma}^{\mathrm{Na}} &= \bigg(1-\frac{2}{m}\bigg(\frac{\Gamma(\frac{m+1}{2})}{\Gamma(\frac{m}{2})}\bigg)^2\bigg) \times \boldsymbol{B}
\end{align}
where we have the notation $\imath = \sqrt{-1}$. By substituting \eqref{Nakagami_1}, \eqref{Nakagami_2} and \eqref{R_s} into \eqref{f_function}, we have the capacity lower bound in the Nakagami-$m$ fading channel as
\begin{equation}
C \ge 2\log L - \mathcal{F}(\sigma, \bm{\mu}^{\mathrm{Na}}, \bm{\varSigma}^{\mathrm{Na}}).
\end{equation}
By substituting \eqref{Nakagami_1}, \eqref{Nakagami_2} and \eqref{R_s} into \eqref{ABER_2}, we have the BER upper bound in the Nakagami-$m$ fading channel as
\begin{equation}
P_\mathrm{e} \le \frac{1}{\pi L\log L} \mathcal{G}(\sigma, \bm{\mu}^{\mathrm{Na}}, \bm{\varSigma}^{\mathrm{Na}}).
\end{equation}
\subsection{BER with Imperfect CSI}
The CSI at the receiver is generally obtained using a pilot sequence and it is demanding to have the accurate CSI. Therefore, we consider the channel model with inaccurate CSI, which is formulated as
\begin{equation}
\boldsymbol{\mathsf{H}} = \hat{\boldsymbol{\mathsf{H}}} + \bm{\upepsilon}
\end{equation}
where $\hat{\boldsymbol{\mathsf{H}}} \in \mathbb{C}^{N \times M}$ denotes the estimation of the true channel; $\bm{\upepsilon} \in \mathbb{C}^{N \times M}$ represents the channel estimation error, with i.i.d. samples following CSCG distribution, i.e., $\bm{\upepsilon} \sim \mathcal{CN}(\boldsymbol{0}, \gamma^2\boldsymbol{I})$. Then, the conditional PEP with imperfect CSI can be expressed as
\begin{align}
\mathrm{Pr}( {\boldsymbol{\mathsf{x}}_i} \to {\boldsymbol{\mathsf{x}}_j} | \hat{\boldsymbol{\mathsf{H}}} ) &= \mathrm{Pr}(\| \boldsymbol{\mathsf{r}}- \hat{\boldsymbol{\mathsf{H}}}\boldsymbol{\mathsf{x}}_i \|^2 > \| \boldsymbol{\mathsf{r}}- \hat{\boldsymbol{\mathsf{H}}}\boldsymbol{\mathsf{x}}_j \|^2 | \hat{\boldsymbol{\mathsf{H}}}) \label{PEP_3}\\
&= \mathrm{Pr}(\| (\boldsymbol{\mathsf{H}}-\hat{\boldsymbol{\mathsf{H}}})\boldsymbol{\mathsf{x}}_i + \boldsymbol{\mathsf{u}} \|^2 > \| \boldsymbol{\mathsf{H}}\boldsymbol{\mathsf{x}}_i - \hat{\boldsymbol{\mathsf{H}}}\boldsymbol{\mathsf{x}}_j + \boldsymbol{\mathsf{u}}\|^2 | \hat{\boldsymbol{\mathsf{H}}})\\
&= \mathrm{Pr}(\| \bm{\upepsilon}\boldsymbol{\mathsf{x}}_i + \boldsymbol{\mathsf{u}} \|^2 > \| \hat{\boldsymbol{\mathsf{H}}}(\boldsymbol{\mathsf{x}}_i-\boldsymbol{\mathsf{x}}_j) + \bm{\upepsilon}\boldsymbol{\mathsf{x}}_i + \boldsymbol{\mathsf{u}} \|^2 | \hat{\boldsymbol{\mathsf{H}}})\\
&= \mathrm{Pr}(\| \boldsymbol{\mathsf{v}} \|^2 > \| \hat{\boldsymbol{\mathsf{H}}}\boldsymbol{\updelta}_{ij} \|^2 - 2\Re(\boldsymbol{\mathsf{v}}^\dagger \hat{\boldsymbol{\mathsf{H}}} \boldsymbol{\updelta}_{ij}) + \| \boldsymbol{\mathsf{v}} \|^2 | \hat{\boldsymbol{\mathsf{H}}})\\
&= \mathrm{Pr}(2\Re(\boldsymbol{\mathsf{v}}^\dagger \hat{\boldsymbol{\mathsf{H}}}\boldsymbol{\updelta}_{ij}) > \| \hat{\boldsymbol{\mathsf{H}}}\boldsymbol{\updelta}_{ij} \|^2 | \hat{\boldsymbol{\mathsf{H}}})\\
&= Q\bigg(\sqrt{\frac{\| \hat{\boldsymbol{\mathsf{H}}}\boldsymbol{\updelta}_{ij} \|^2}{2\eta^2}}\bigg) \label{PEP_4}
\end{align}
where we have the shorthand $\boldsymbol{\mathsf{v}} = \bm{\upepsilon}\boldsymbol{\mathsf{x}}_i+\boldsymbol{\mathsf{u}}$ and $\eta^2 = \gamma^2 \|\boldsymbol{\mathsf{x}}_i\|^2 + \sigma^2$. The equations \eqref{PEP_3}-\eqref{PEP_4} follows from mathematical manipulations and the fact that the equivalent noise $\boldsymbol{\mathsf{v}}$ follows CSCG distribution. Let $\bar{\bm{\varSigma}} = \bm{\varSigma} + \gamma^2 \boldsymbol{I}$, the general BER upper bound in the scenario of imperfect channel estimation is expressed as
\begin{equation}\label{ABER_3}
P_\mathrm{e} \le \frac{1}{\pi L\log L} \mathcal{G}(\eta, \bm{\mu}, \bm{\bar{\varSigma}}).
\end{equation}
\begin{remark}
The imperfect channel estimation incurs an error floor of the BER performance, which is due to the introduction of an additional variance term $\gamma^2$ associated with the channel estimation error. At high SNR, the noise variance $\sigma^2$ approaches to zero while $\eta^2$ approaches to $\gamma^2 \|\boldsymbol{\mathsf{x}}_i\|^2$, which is always larger than zero and thus keeps the BER at a fixed level. In addition, it is also readily to know that the BER floor increases as the error variance increases.
\end{remark}

\section{Numerical Results}\label{Simulations}
In this section, we provide the numerical results on the capacity and the reliability of the proposed SCM scheme in various channel scenarios.
\begin{figure*}[!t]
\begin{minipage}[t]{0.49\linewidth}
  \centering
  \includegraphics[width=3in]{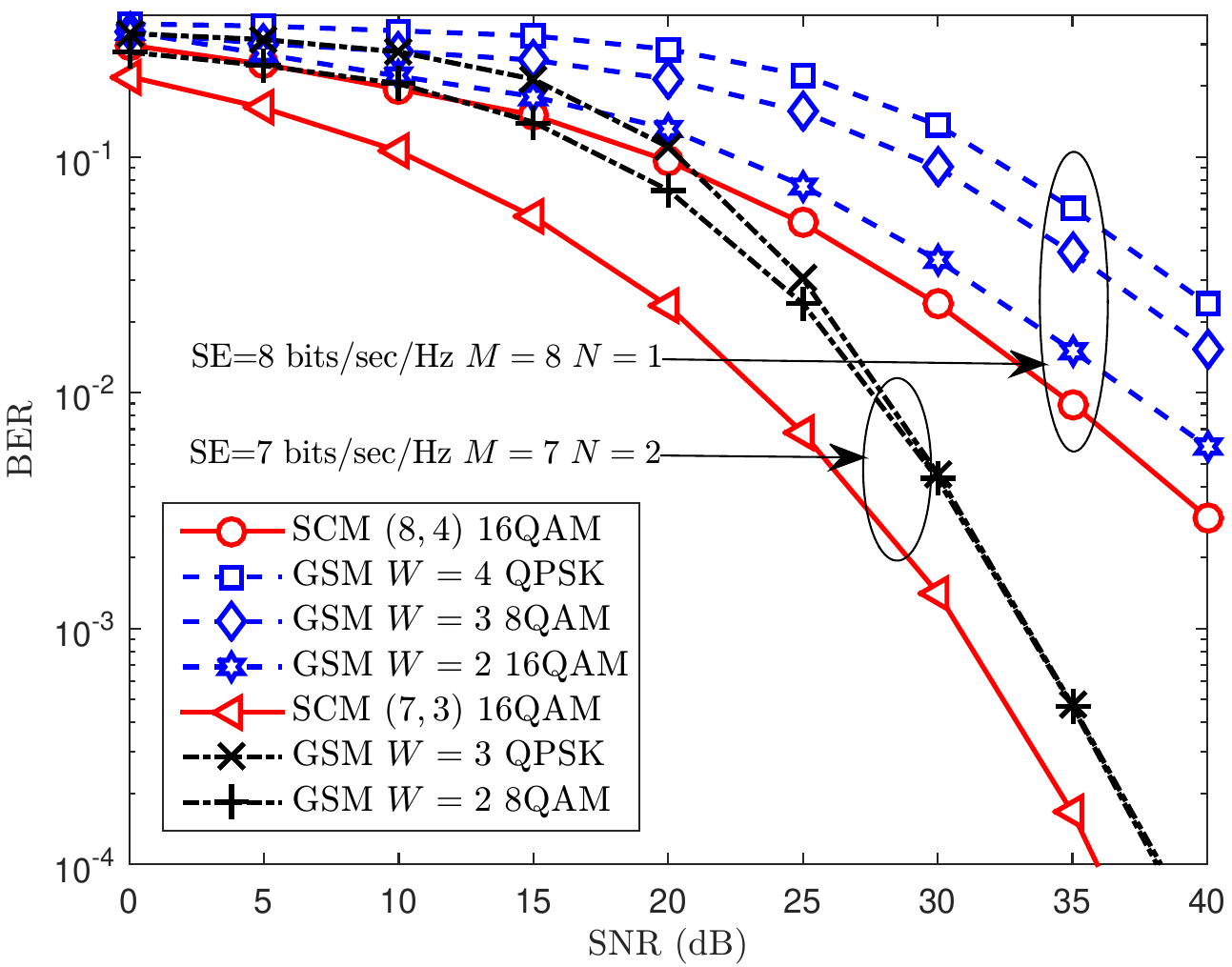}\\
  \caption{The BER performance of the SCM and the GSM in the spatially correlated Rician channel, where the Rician $K$ factor is $7$; the spatial correlation coefficients at the transmitter and the receiver are $\rho = \tau = 0.5$. The receiver has perfect CSI.}
  \label{Compr_BER_versus_SNR}
  \end{minipage}%
  \hfill
\begin{minipage}[t]{0.49\linewidth}
  \centering
  \includegraphics[width=3in]{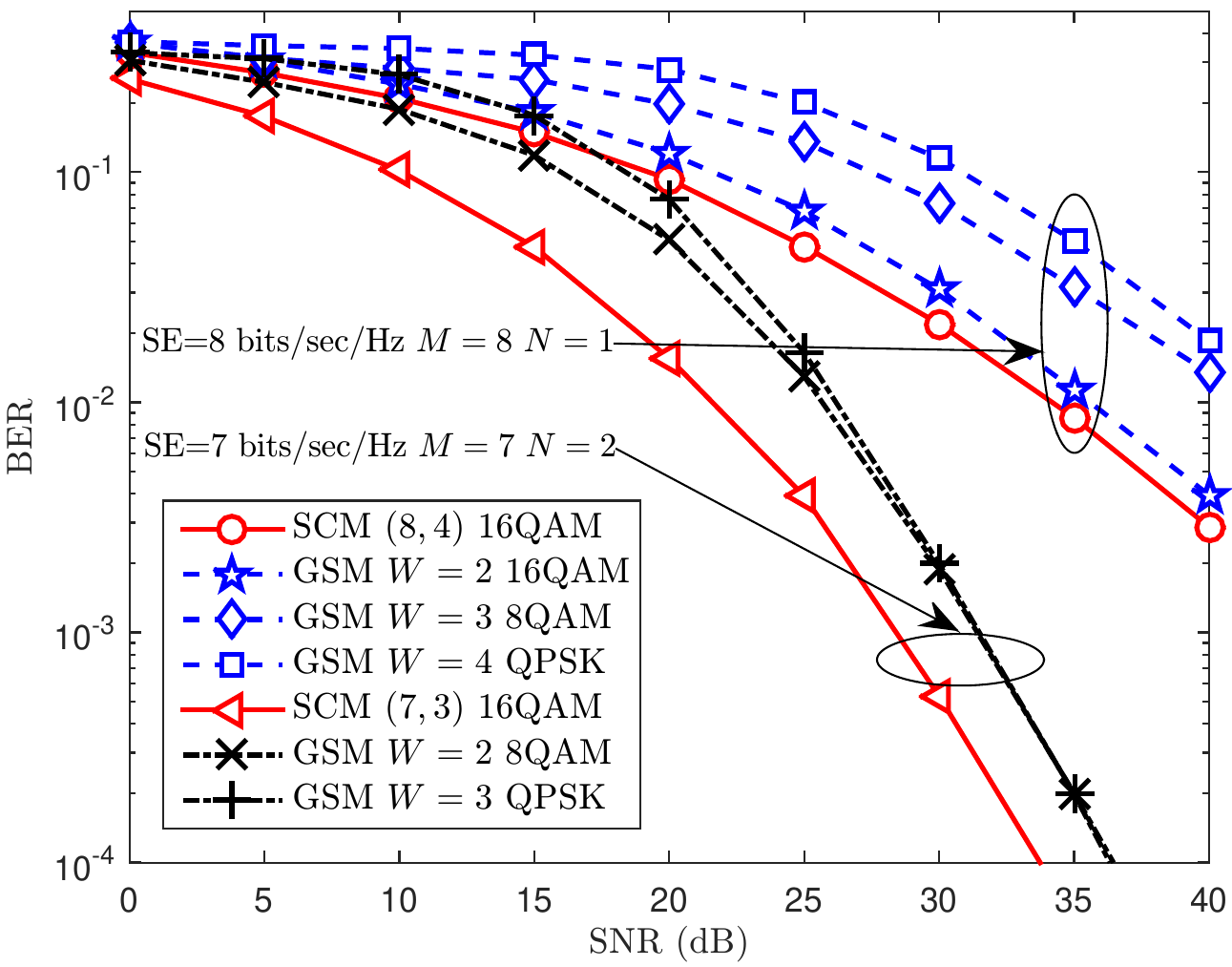}\\
  \caption{The BER performance of the SCM and the GSM in the spatially correlated Nakagami-$m$ channel, where the $m$ factor is $3$; the spatial correlation coefficients at the transmitter and the receiver are $\rho = \tau = 0.2$. The receiver has perfect CSI.}
  \label{Compr_BER_versus_SNR_Nakagami}
  \end{minipage}
\end{figure*}



\noindent\textit{\textbf{Observation 1}: For an identical antenna configuration, the proposed SCM outperforms the conventional GSM in the reliability (cf. Figs. \ref{Compr_BER_versus_SNR} and \ref{Compr_BER_versus_SNR_Nakagami}).}

In Figs. $\ref{Compr_BER_versus_SNR}$ and $\ref{Compr_BER_versus_SNR_Nakagami}$, the BER performance of the proposed SCM scheme is compared with the conventional GSM in the spatially correlated Rician and Nakagami-$m$ channels. Targeting at a spectral efficiency of $7$ bits/sec/Hz, we compare the SCM $(7,4)$ and the SCM $(7,3)$ with the conventional GSM having $W = 2, 3$ active antennas. For a spectral efficiency of $8$ bits/sec/Hz, we compare the SCM $(8,4)$ with the conventional GSM having $W = 2, 3, 4$ active antennas. The numerical results show that the proposed SCM remarkably outperforms the conventional GSM. Specifically, compared with the conventional GSM scheme in Fig. \ref{Compr_BER_versus_SNR}, the SCM $(7,3)$ achieves $2.5$ dB gain at the BER of $0.0001$ for the spectral efficiency of $7$ bits/sec/Hz and the SCM $(7,4)$ achieves $4$ dB gain at the BER of $0.01$ for the spectral efficiency of $8$ bits/sec/Hz. Similar trend is observed in Fig. \ref{Compr_BER_versus_SNR_Nakagami}, where the channel follows Nakagami-$m$ fading. The reliability enhancement is due to the increasing of the minimum Hamming distance, which has been proved in Section \ref{Performance_Analysis}. In particular, the minimum Hamming distances of the SCM $(7,3)$, the SCM $(7,4)$ and the SCM $(8,4)$ are $4$, $3$ and $4$, respectively, while that of the conventional GSM schemes are all fixed at $2$. In contrast to the conventional coding strategies which suffer from redundancy and efficiency, the proposed coding-over-antenna approach improves the reliability without sacrificing the communication efficiency.

\begin{figure*}[!t]
\begin{minipage}[t]{0.49\linewidth}
  \centering
  \includegraphics[width=2.9in]{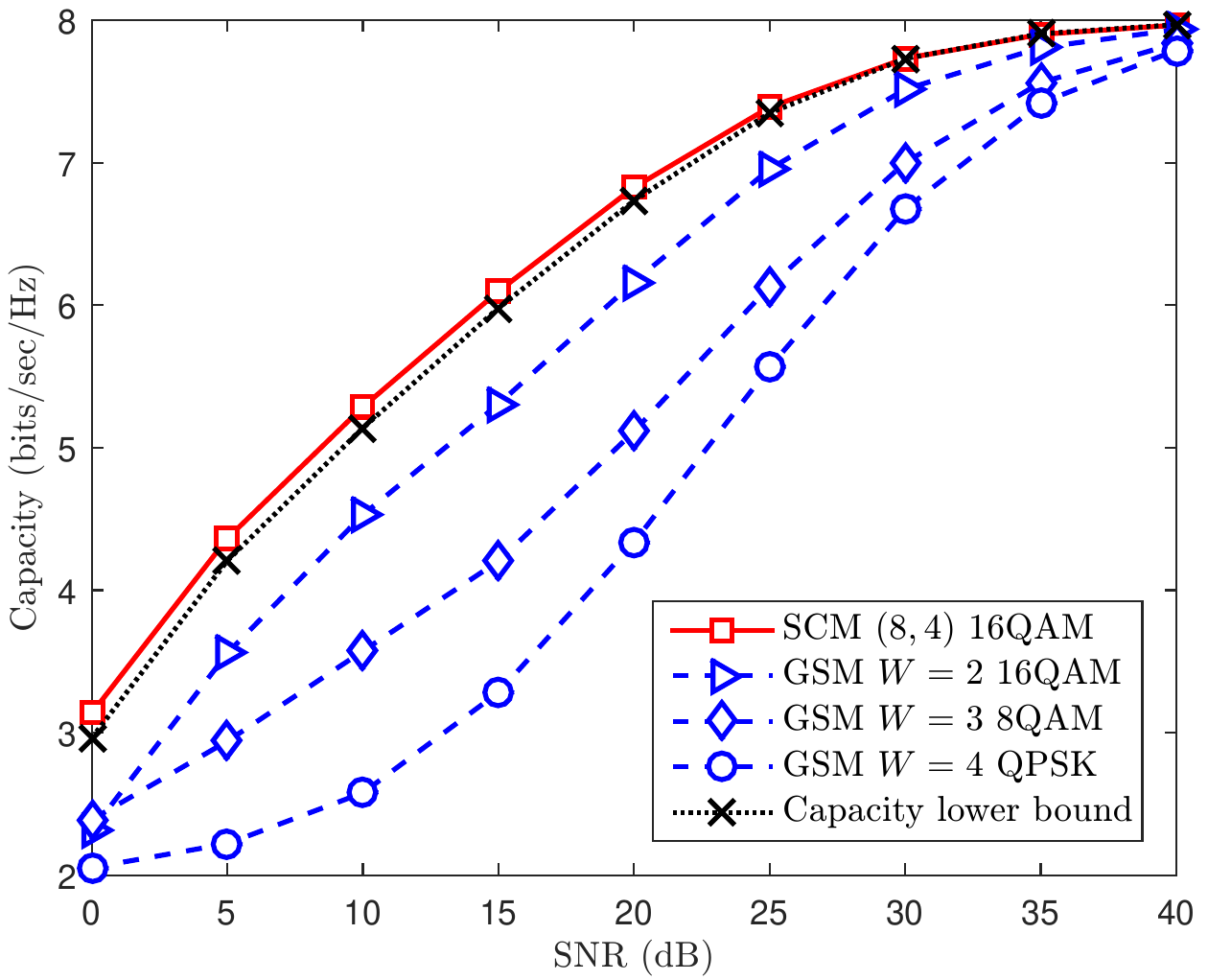}\\
  \caption{The capacity of the SCM $(8,4)$ and the GSM in the spatially correlated Rician channel, where the Rician $K$ factor is $7$; the numbers of the transmit and receive antennas are $M = 8$ and $N = 1$, respectively; the transmitter spatial correlation coefficient is $\tau = 0.5$.}
  \label{Compr_DCMC_versus_SNR}
  \end{minipage}%
  \hfill
\begin{minipage}[t]{0.49\linewidth}
  \centering
  \includegraphics[width=2.9in]{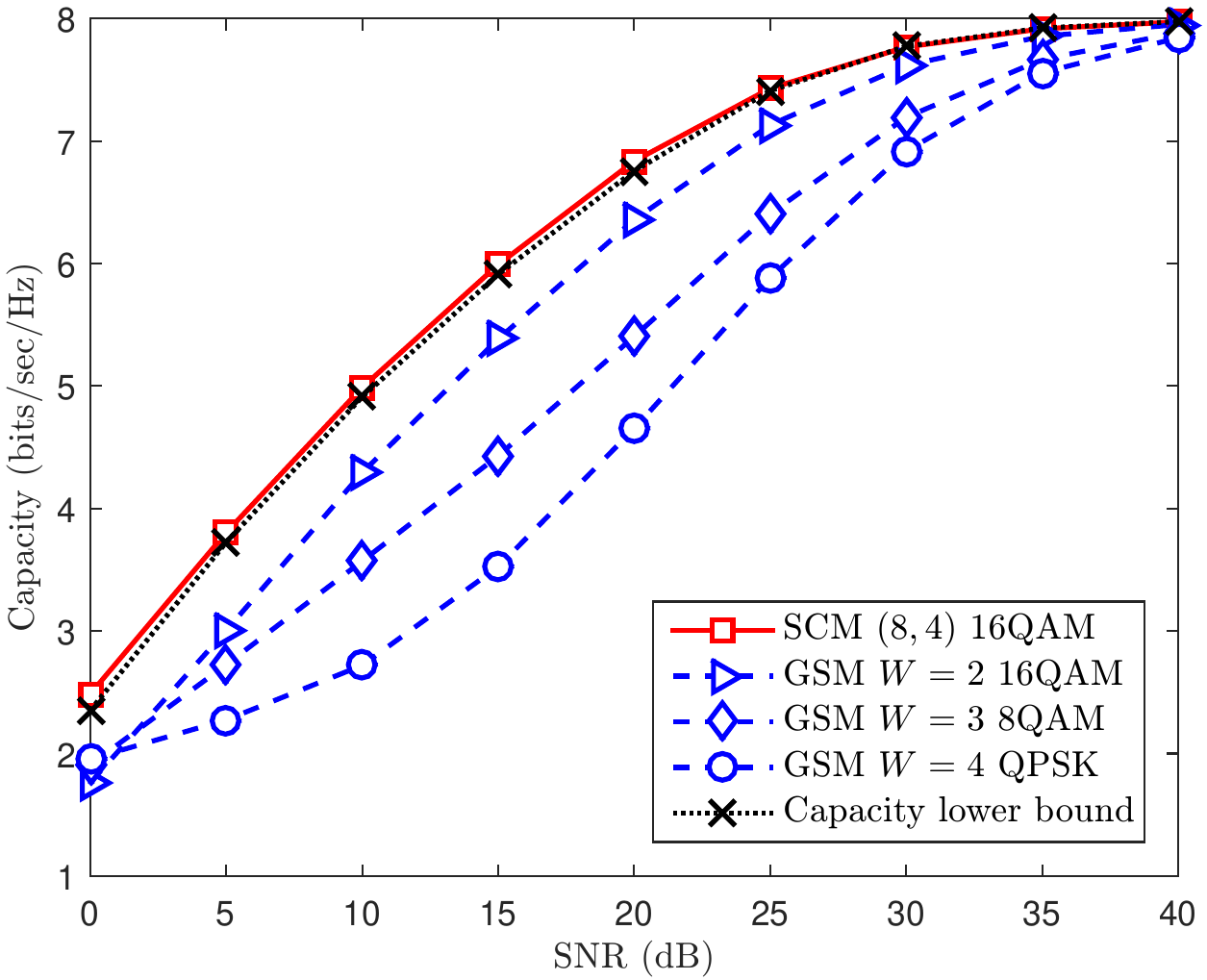}\\
  \caption{The capacity of the SCM $(8,4)$ and the GSM in the spatially correlated Nakagami-$m$ channel, where the $m$ factor is $3$; the numbers of the transmit and receive antennas are $M = 8$ and $N = 1$, respectively; the transmitter spatial correlation coefficient is $\tau = 0.2$.}
  \label{Compr_DCMC_versus_SNR_Nakagami}
  \end{minipage}
\end{figure*}



\noindent\textit{\textbf{Observation 2}: For an identical antenna configuration, the proposed SCM outperforms the conventional GSM in the capacity (cf. Figs. \ref{Compr_DCMC_versus_SNR} and \ref{Compr_DCMC_versus_SNR_Nakagami}).}

In Figs. $\ref{Compr_DCMC_versus_SNR}$ and $\ref{Compr_DCMC_versus_SNR_Nakagami}$, we evaluate the capacities of the SCM $(8,4)$ and the GSM in the spatially correlated Rician and Nakagami-$m$ channels, in which the channel scenarios and antenna configurations are identical to that in Figs. $\ref{Compr_BER_versus_SNR}$ and $\ref{Compr_BER_versus_SNR_Nakagami}$. From Figs. $\ref{Compr_DCMC_versus_SNR}$ and $\ref{Compr_DCMC_versus_SNR_Nakagami}$, we observe that: First, both the SCM $(8,4)$ and the conventional GSM schemes with $W = 2, 3, 4$ achieve the same ultimate capacity of $8$ bits/sec/Hz at sufficiently high SNR; Second, the SCM $(8,4)$ exhibits a higher capacity than the conventional GSM schemes before reaching the $8$ bits/sec/Hz value. The capacity improvement is attributed to the increasing of the minimum Hamming distance, which has been proved in Section \ref{Performance_Analysis}. Therefore, we conclude that both higher capacity and better reliability can be achieved by the proposed SCM.

\begin{figure*}[!t]
\begin{minipage}[t]{0.49\linewidth}
  \centering
  \includegraphics[width=3in]{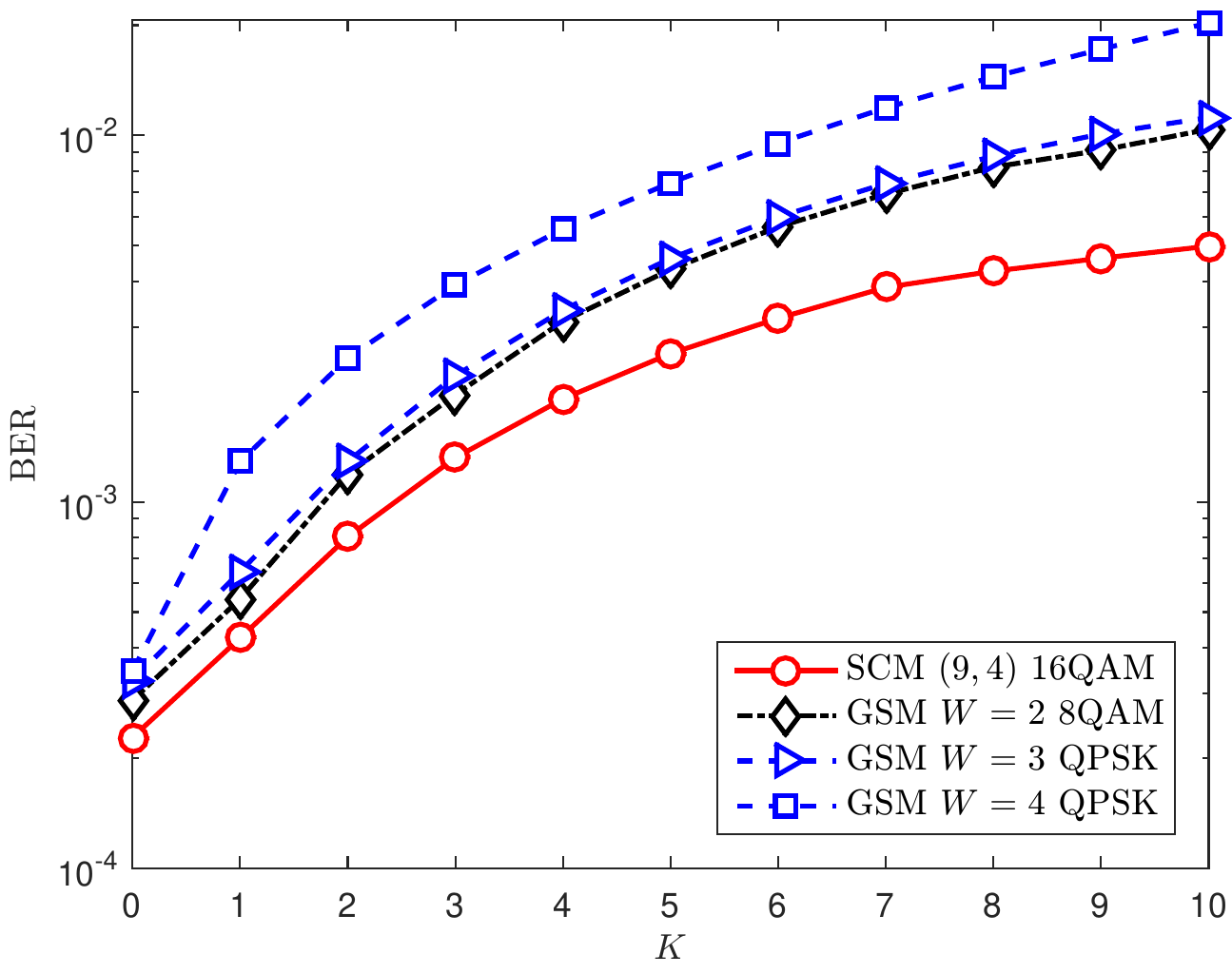}\\
  \caption{The BER performance of the SCM $(9,4)$ and the GSM in the spatially correlated Rician channel, where the SNR is $30$ dB; the numbers of the transmit antennas and the receive antennas are $M = 9$ and $N = 2$, respectively; the spatial correlation coefficients at the transmitter and the receiver are $\rho = \tau = 0.5$. The system targets a spectral efficiency of $8$ bits/sec/Hz and the perfect CSI is assumed at the receiver.}
  \label{Compr_BER_versus_K}
  \end{minipage}%
  \hfill
\begin{minipage}[t]{0.49\linewidth}
  \centering
  \includegraphics[width=3in]{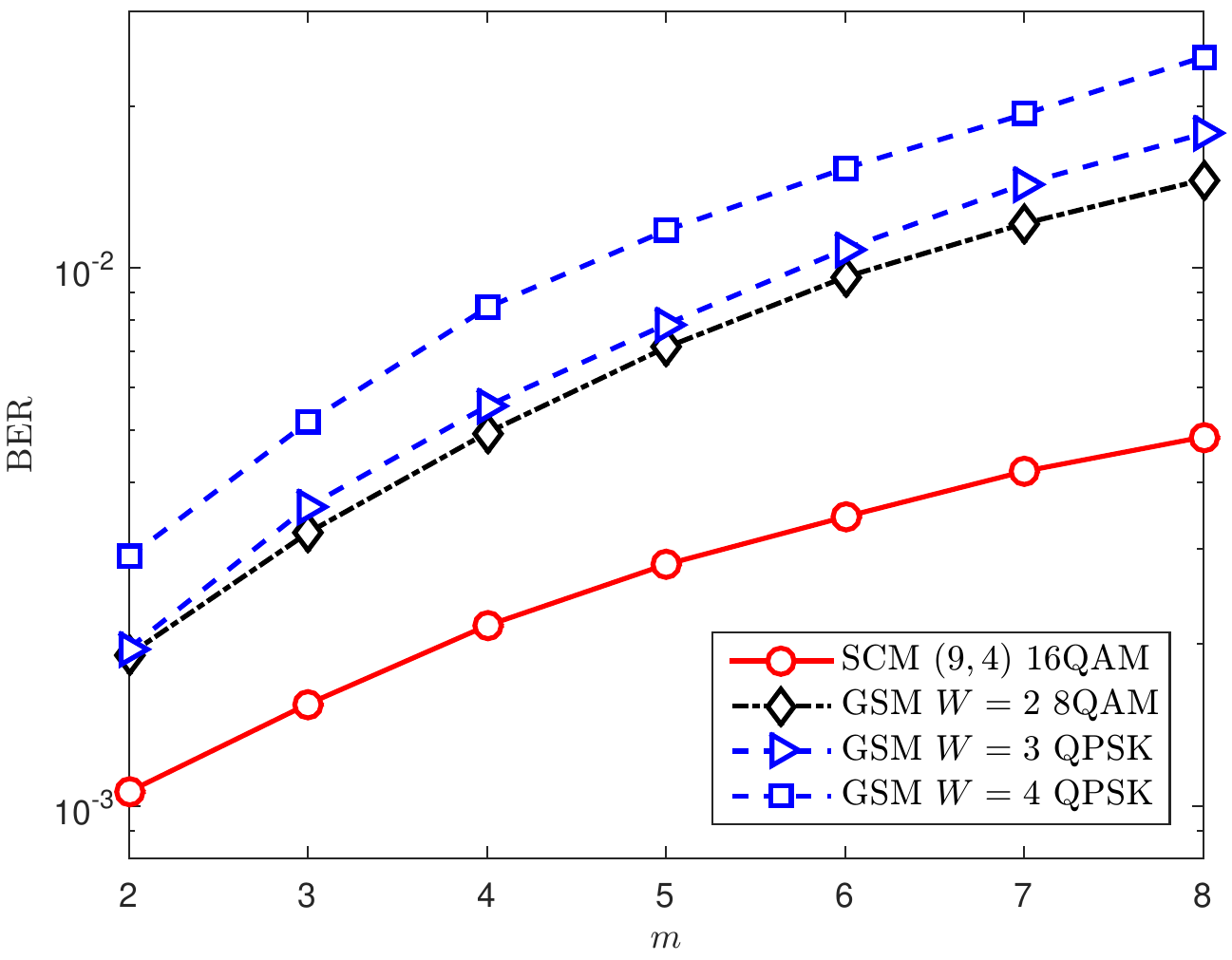}\\
  \caption{The BER performance of the SCM $(9,4)$ and the GSM in the spatially correlated Nakagami-$m$ channel, where the SNR is $30$ dB; the numbers of the transmit antennas and the receive antennas are $M = 9$ and $N = 2$, respectively; the spatial correlation coefficients at the transmitter and the receiver are $\rho = \tau = 0.2$. The system targets a spectral efficiency of $8$ bits/sec/Hz and the perfect CSI is assumed at the receiver.}
  \label{Compr_BER_versus_m}
  \end{minipage}
\end{figure*}



\begin{figure*}[!t]
\begin{minipage}[t]{0.49\linewidth}
  \centering
  \includegraphics[width=3in]{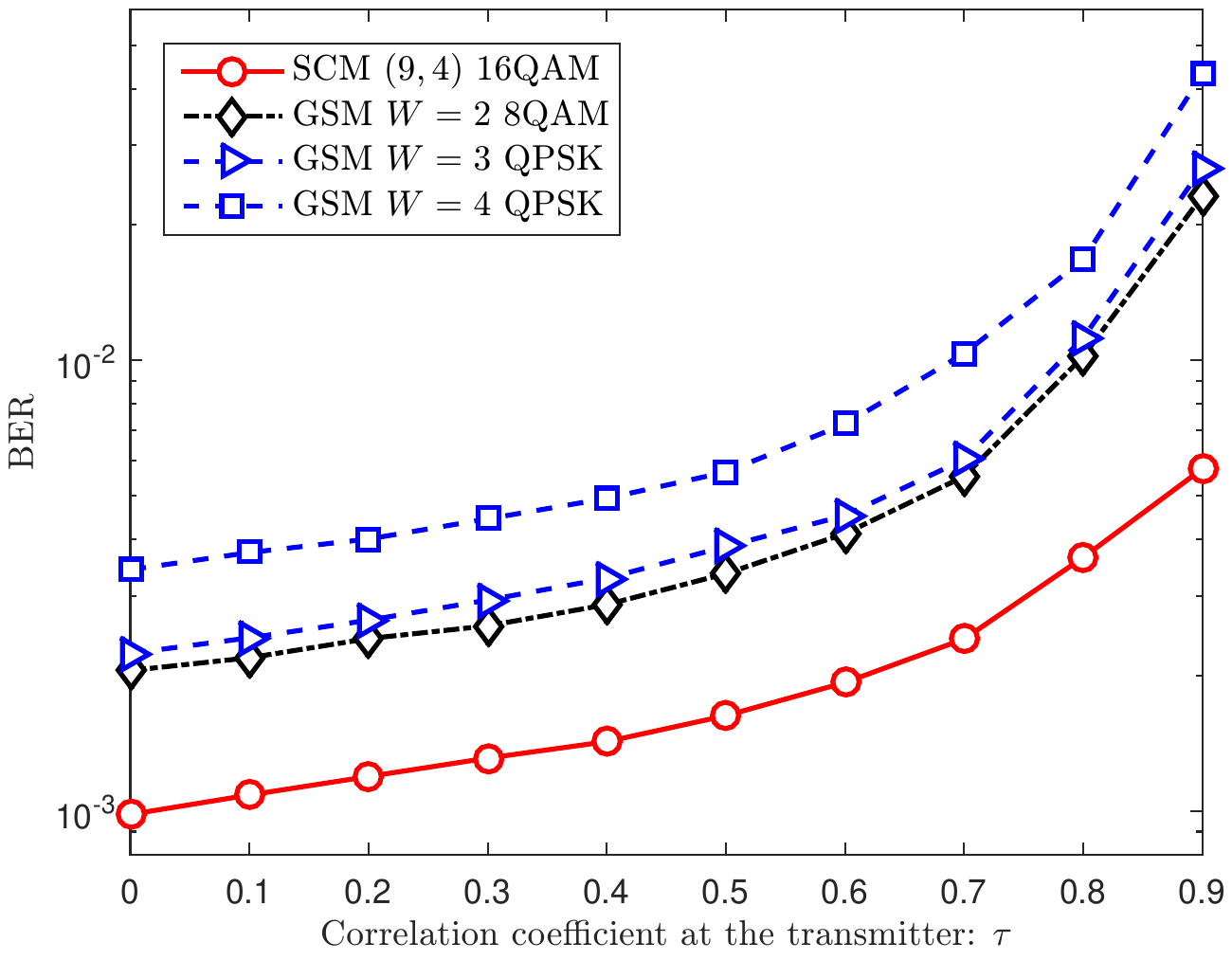}\\
  \caption{The BER performance of the SCM $(9,4)$ and the GSM in the spatially correlated Rician channel, where the SNR is $30$ dB; the numbers of the transmit antennas and the receive antennas are $M = 9$ and $N = 2$, respectively; the Rician $K$ factor is $6$. The system targets a spectral efficiency of $8$ bits/sec/Hz and the perfect CSI is assumed at the receiver.}
  \label{Compr_BER_versus_rhot}
  \end{minipage}%
  \hfill
\begin{minipage}[t]{0.49\linewidth}
  \centering
  \includegraphics[width=3in]{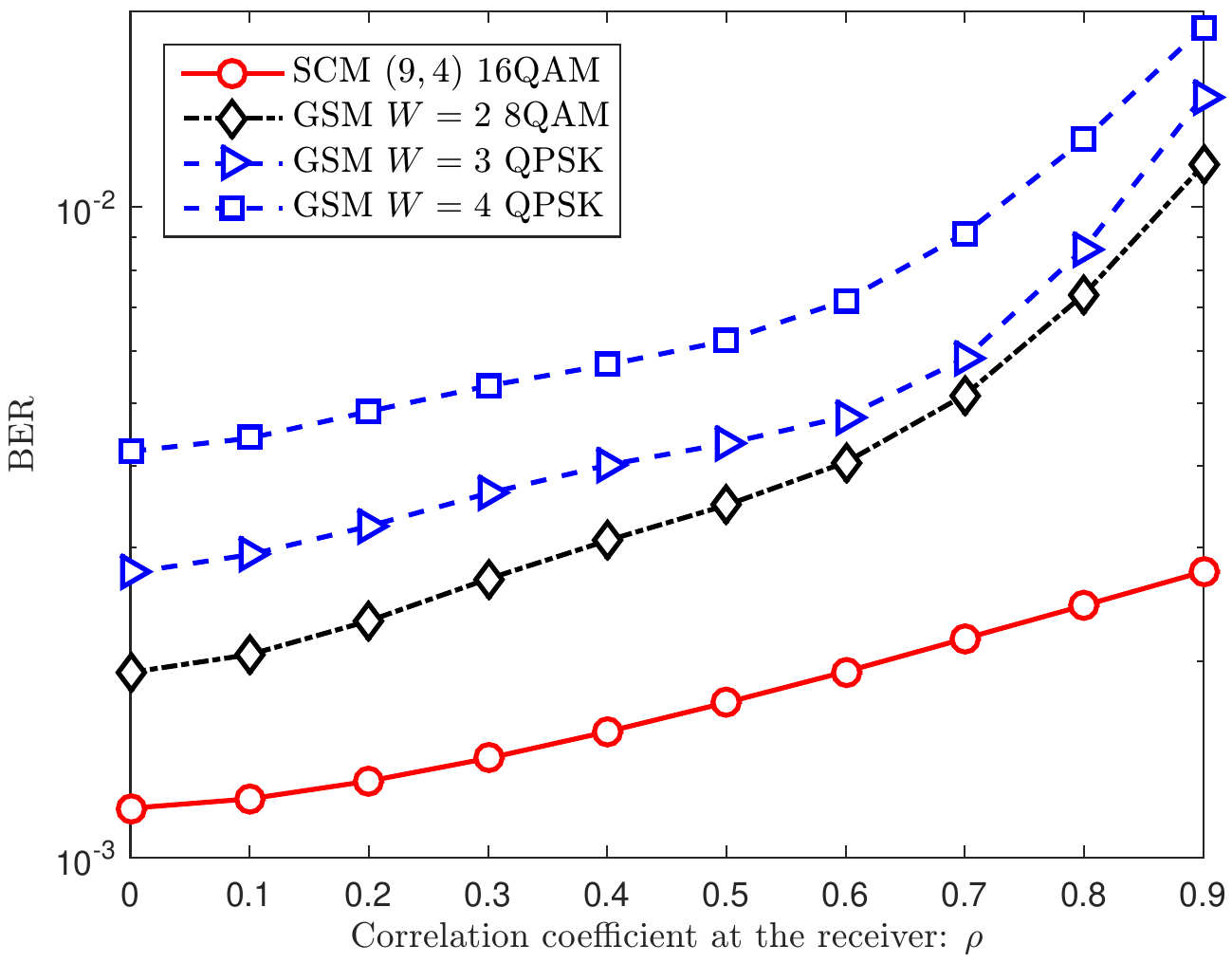}\\
  \caption{The BER performance of the SCM $(9,4)$ and the GSM in the spatially correlated Rician channel, where the SNR is $30$ dB; the numbers of the transmit antennas and the receive antennas are $M = 9$ and $N = 2$, respectively; the Rician $K$ factor is $5$. The system targets a spectral efficiency of $8$ bits/sec/Hz and the perfect CSI is assumed at the receiver.}
  \label{Compr_BER_versus_rhor}
  \end{minipage}
\end{figure*}



\noindent\textit{\textbf{Observation 3}: The performance gain of the proposed SCM over the conventional GSM increases as the channel correlation becomes more severe (cf. Figs. \ref{Compr_BER_versus_K}-\ref{Compr_BER_versus_rhor}).}

The reliability improvement is evaluated in various channel scenarios, i.e., the Rician $K$ factor, the Nakagami $m$ factor, the spatial correlation coefficients $\rho$ and $\tau$. In Figs. $\ref{Compr_BER_versus_K}$, $\ref{Compr_BER_versus_m}$, $\ref{Compr_BER_versus_rhot}$ and $\ref{Compr_BER_versus_rhor}$, the BER performance of the proposed SCM $(9,4)$ is compared with the conventional GSM schemes with $W = 2, 3, 4$. All the transmission schemes have $9$ transmit antennas and $2$ receive antennas and target at a spectral efficiency of $8$ bits/sec/Hz. The numerical results show that the proposed SCM scheme outweighs the conventional GSM schemes in various channel scenarios and the performance gain increases as the channel correlation becomes stronger. Specifically, the proposed SCM scheme has lower BER than the conventional GSM schemes when the Rician $K$ factor ranges from $0$ to $10$, and the Nakagami $m$ factor ranges from $2$ to $8$. It is worth noting that the Rician fading reduces to the Rayleigh fading with $K = 0$, which means that the proposed SCM is also able to provide slight performance gains in the Rayleigh fading channel. In addition, from Figs. $\ref{Compr_BER_versus_rhot}$ and $\ref{Compr_BER_versus_rhor}$, we observe that the proposed SCM achieves better reliability than the conventional GSM schemes when the channels are spatially correlated. Moreover, the performance gap between the SCM and the GSM becomes larger as the channels are more scarcely scattered. This is due to that, in these channel scenarios, the identification of the active antenna activation patterns are even more unreliable due to similar channel states from transmit antenna groups to the receive antennas. The proposed SCM directly encodes over the transmit antennas, which improves the accuracy of the active antenna detection and yields better reliability.

\noindent\textit{\textbf{Observation 4}: The determined capacity lower bound and the BER upper bound match well with the numerical results at high SNR. In addition, the BER upper bound can accurately characterize the BER floor in the presence of channel uncertainties (cf. Figs. \ref{Compr_DCMC_versus_SNR}, \ref{Compr_DCMC_versus_SNR_Nakagami}, \ref{Compr_SimBER_versus_AnaBER}-\ref{Compr_SimBER_versus_AnaBER_CSE_over_Nakagami}).}



\begin{figure*}[!t]
\begin{minipage}[t]{0.49\linewidth}
  \centering
  \includegraphics[width=3in]{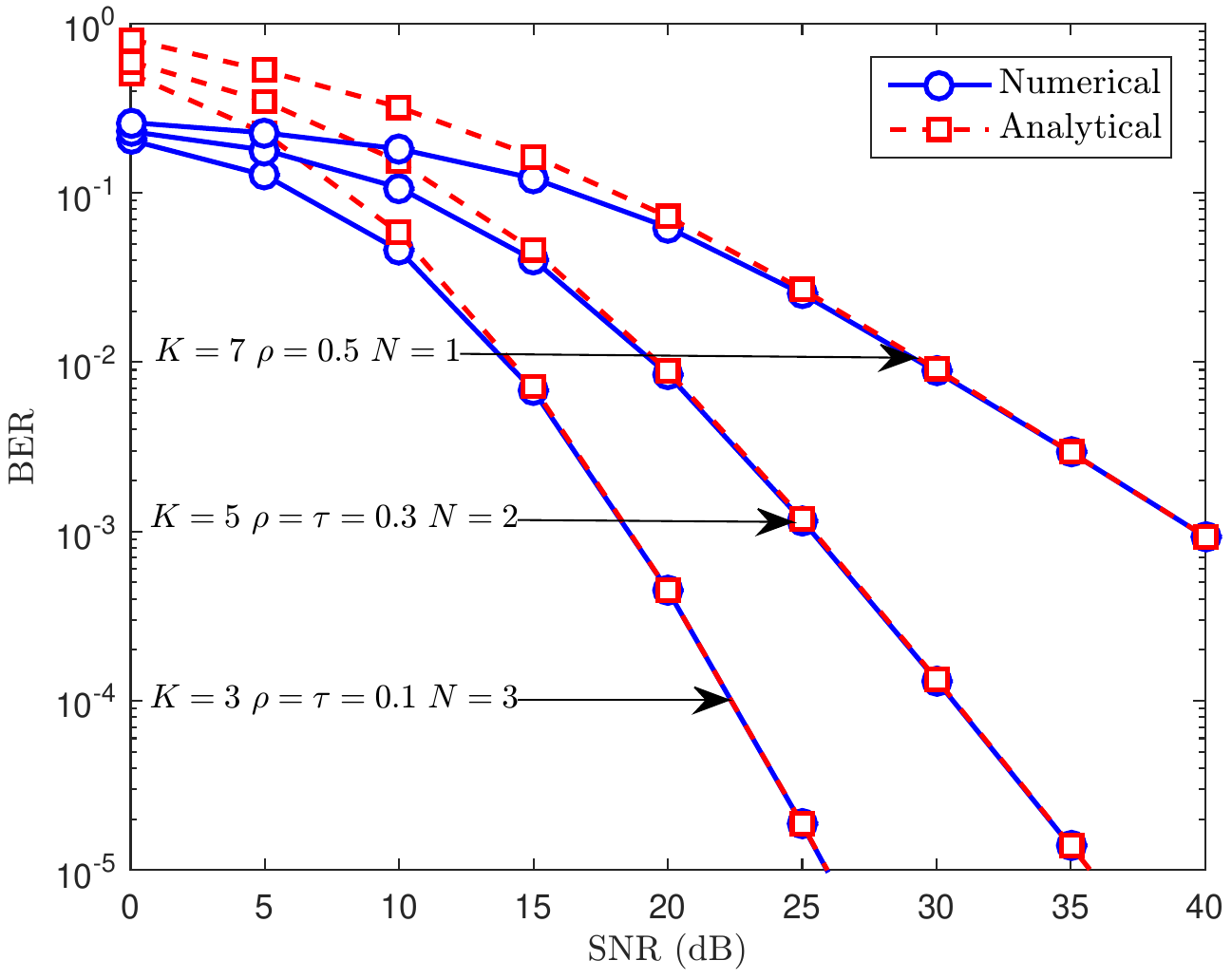}\\
  \caption{The BER performance of the SCM $(7,3)$ in the spatially correlated Rician channel, where the modulation format is QPSK. The results show that the derived BER upper bound is tight at high SNR especially when SNR $> 20$ dB.}
  \label{Compr_SimBER_versus_AnaBER}
  \end{minipage}%
  \hfill
\begin{minipage}[t]{0.49\linewidth}
  \centering
  \includegraphics[width=3in]{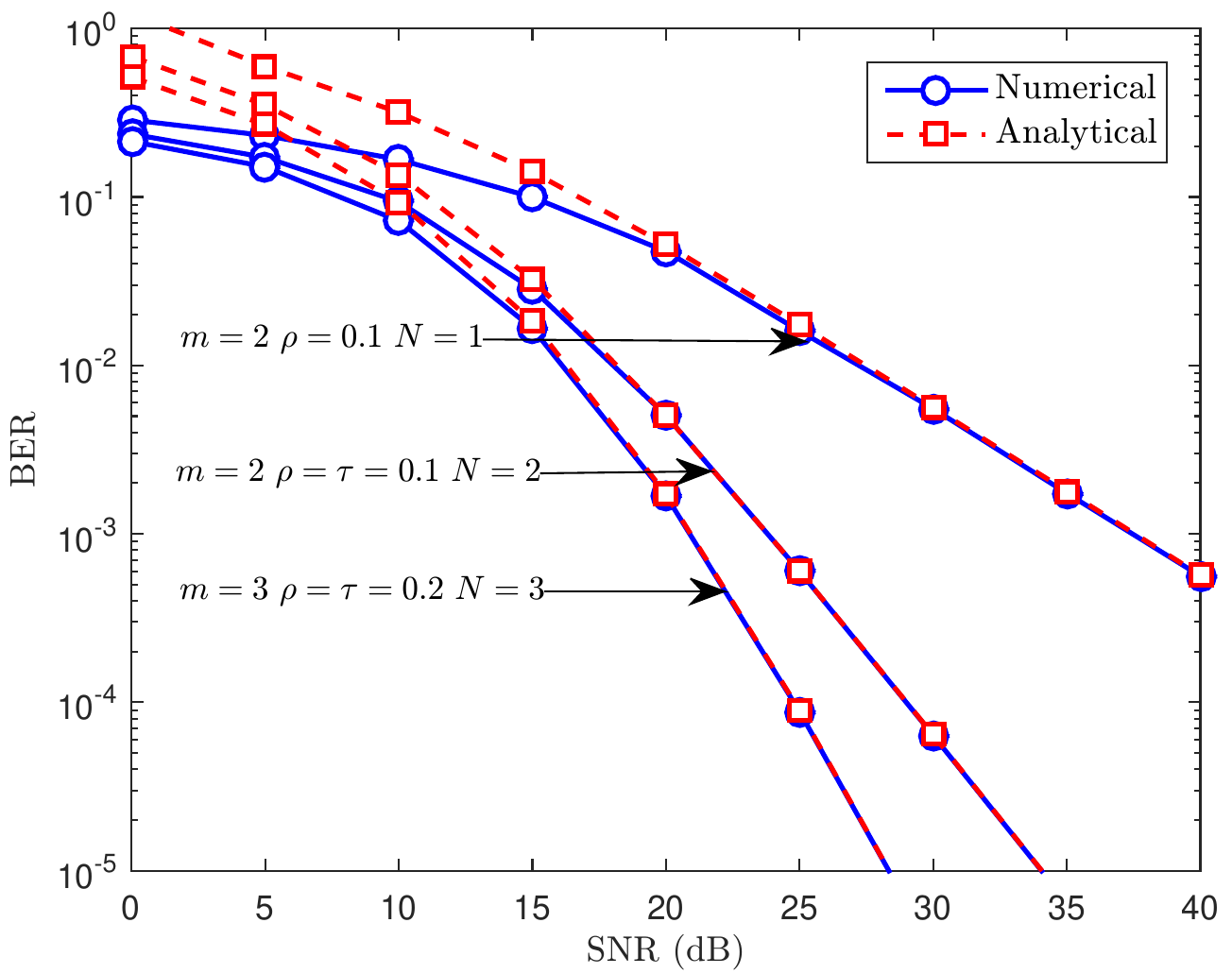}\\
  \caption{The BER performance of the SCM $(7,3)$ in the spatially correlated Nakagami-$m$ channel, where the modulation format is QPSK. The results show that the derived BER upper bound is tight at high SNR especially when SNR $> 20$ dB.}
  \label{Compr_SimBER_versus_AnaBER_over_Nakagami}
  \end{minipage}
\end{figure*}

From Figs. $\ref{Compr_DCMC_versus_SNR}$ and $\ref{Compr_DCMC_versus_SNR_Nakagami}$, we observe that the derived capacity lower bound is tight when SNR is larger than $25$ dB. In addition, in Figs. $\ref{Compr_SimBER_versus_AnaBER}$ and $\ref{Compr_SimBER_versus_AnaBER_over_Nakagami}$, it depicts both the numerical BER performance and the BER upper bound for the SCM $(7,3)$ with the QPSK modulation. The perfect CSI is assumed at the receiver. In Fig. $\ref{Compr_SimBER_versus_AnaBER}$, the BER performance is evaluated in the Rician fading and the Nakagami-$m$ fading is employed in Fig. $\ref{Compr_SimBER_versus_AnaBER_over_Nakagami}$. The accuracy of the analytical BER upper bound is validated with various channel parameters, i.e., the Rician $K$ factor, the Nakagami $m$ factor, the spatial correlation coefficients $\rho$ and $\tau$. The comparisons reveal that the proposed BER upper bound is tight at high SNR especially when SNR $> 20$ dB, which further substantiates the theoretical derivations in Section \ref{Performance_Analysis}.



\begin{figure*}[!t]
\begin{minipage}[t]{0.49\linewidth}
  \centering
  \includegraphics[width=3in]{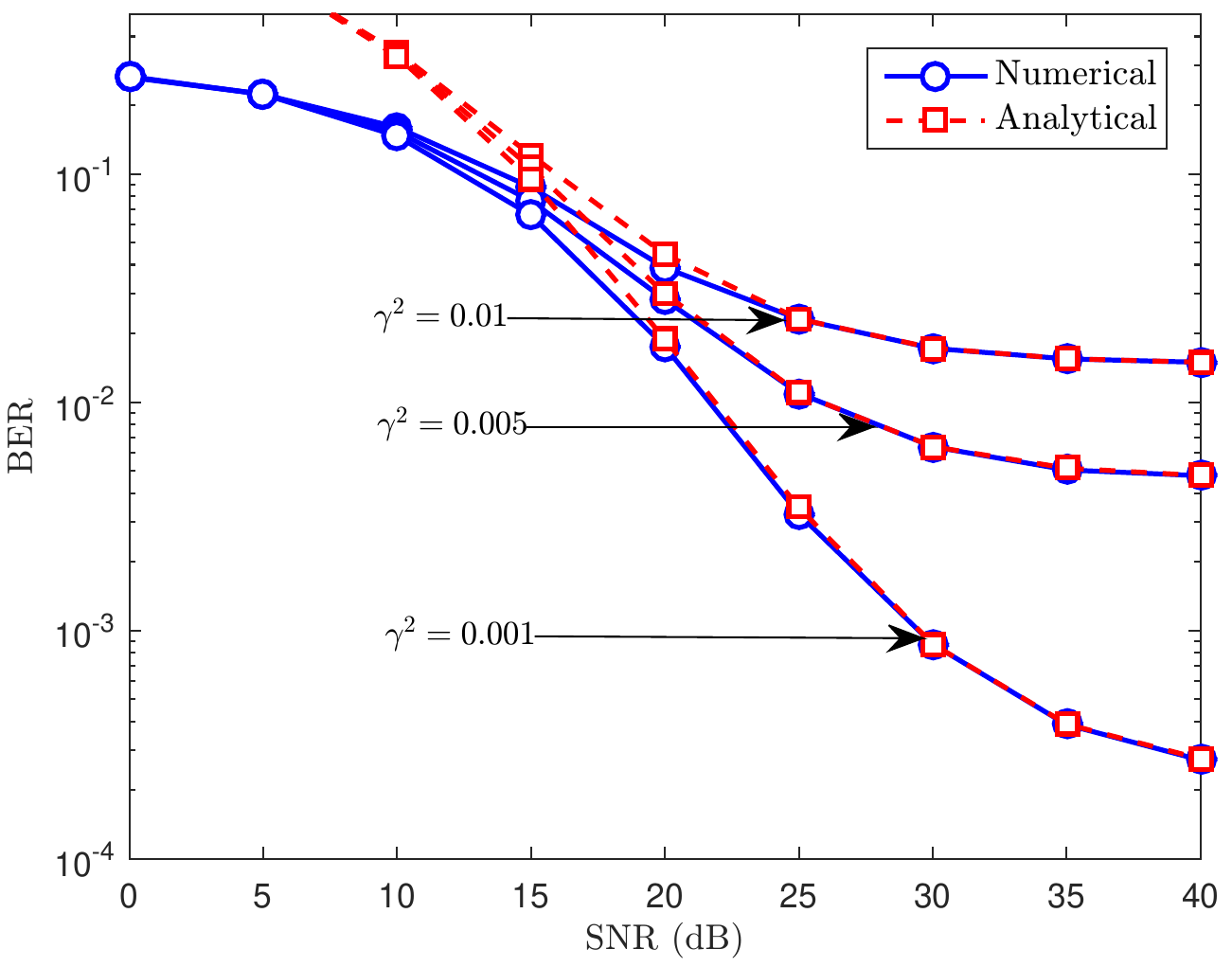}\\
  \caption{The BER performance of the SCM $(7,4)$ in the spatially correlated Rician channel, where the modulation format is QPSK. The Rician $K$ factor is $5$; the numbers of the transmit antennas and the receive antennas are $M = 7$ and $N = 2$, respectively; the spatial correlation coefficients at the transmitter and the receiver are $\rho = \tau = 0.3$. The proposed BER upper bound can accurately characterize the BER floor.}
  \label{Compr_SimBER_versus_AnaBER_CSE}
  \end{minipage}%
  \hfill
\begin{minipage}[t]{0.49\linewidth}
  \centering
  \includegraphics[width=3in]{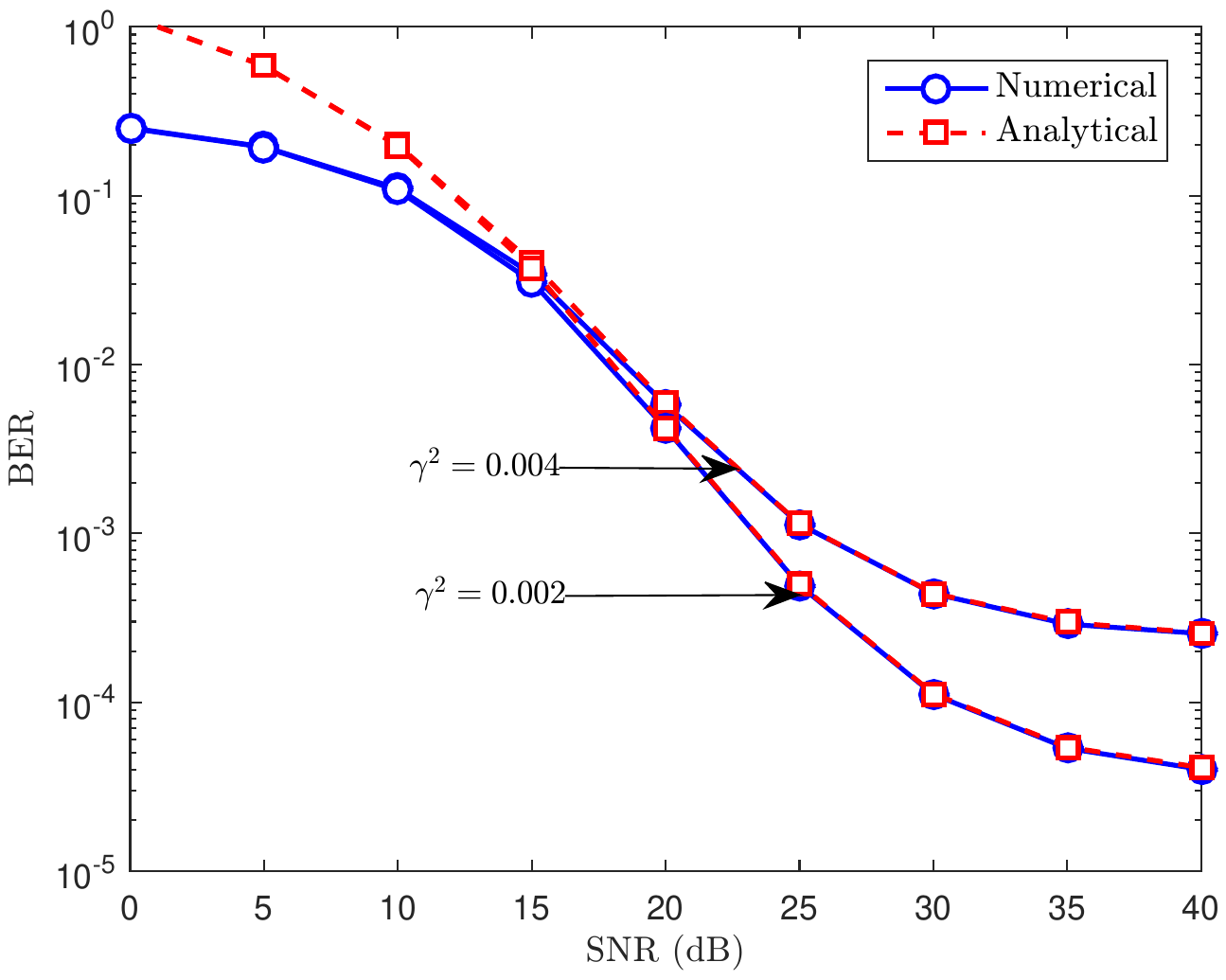}\\
  \caption{The BER performance of the SCM $(7,4)$ in the spatially correlated Nakagami-$m$ channel, where the modulation format is QPSK. The Nakagami $m$ factor is $3$; the numbers of the transmit antennas and the receive antennas are $M = 7$ and $N = 3$, respectively; the spatial correlation coefficients at the transmitter and the receiver are $\rho = \tau = 0.2$. The proposed BER upper bound can accurately characterize the BER floor.}
  \label{Compr_SimBER_versus_AnaBER_CSE_over_Nakagami}
  \end{minipage}
\end{figure*}

In Figs. $\ref{Compr_SimBER_versus_AnaBER_CSE}$ and $\ref{Compr_SimBER_versus_AnaBER_CSE_over_Nakagami}$, we compare the numerical BER performance with the BER upper bound with different channel uncertainties, i.e., $\gamma^2 = 0.01, 0.005, 0.004, 0.002, 0.001$. In Fig. $\ref{Compr_SimBER_versus_AnaBER_CSE}$, the BER performance are evaluated in the Rician fading and the Nakagami-$m$ fading is employed in Fig. $\ref{Compr_SimBER_versus_AnaBER_CSE_over_Nakagami}$. From Figs. $\ref{Compr_SimBER_versus_AnaBER_CSE}$ and $\ref{Compr_SimBER_versus_AnaBER_CSE_over_Nakagami}$, we observe that: First, the BER upper bound with imperfect CSI is tight when SNR is larger than $20$ dB; Second, the BER floor levels in different channel uncertainties can be accurately characterized by the analytical BER upper bound; Third, the error floor of the proposed SCM rises up as the channel estimation error increases. Specifically, the SCM $(7,4)$ has an error floor of $0.005$ when the variance of the channel estimation error is $0.005$, while it floors at $0.015$ when the variance of the channel estimation error is $0.01$.

\noindent\textit{\textbf{Observation 5}: The suboptimal detector strikes a flexible balance between reliability and computational complexity and achieves near-optimal performance by having $3$ antenna activation pattern candidates (cf. Fig. \ref{Compr_SVD_versus_ML}).}

\begin{figure}[!t]
  \centering
  \includegraphics[width=3in]{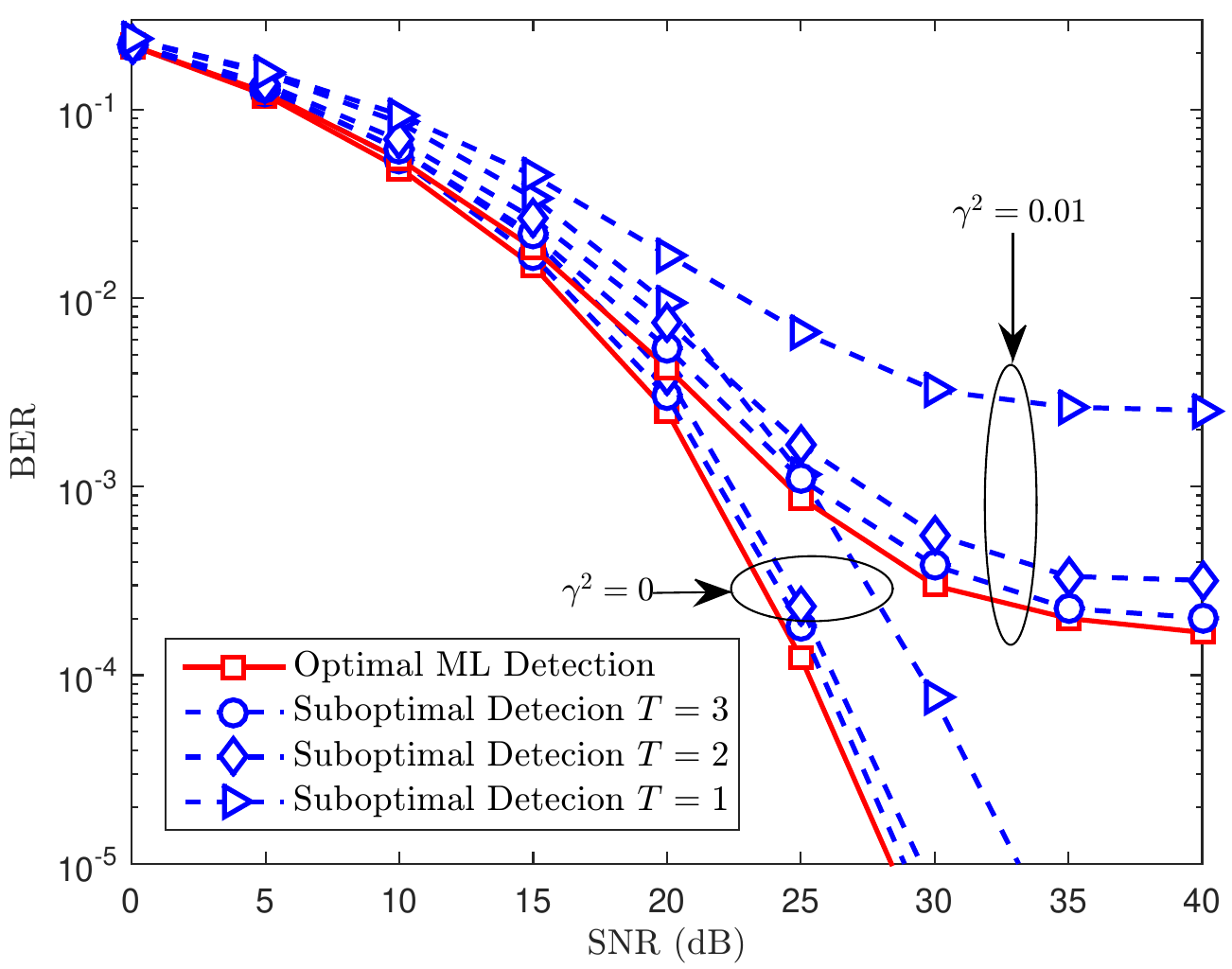}\\
  \caption{The BER performance of the the proposed suboptimal detector and the optimal ML detector in the Rician channel. The modulation format of the SCM $(7,3)$ is $16$ QAM; the numbers of the transmit antennas and the receive antennas are $M = 7$ and $N = 4$, respectively; the Rician $K$ factor is $5$. $T$ is the cardinality of the candidate set $\mathcal{T}$.}
  \label{Compr_SVD_versus_ML}
\end{figure}

In Fig. $\ref{Compr_SVD_versus_ML}$, it depicts the BER performance of the optimal and the suboptimal detectors for the SCM $(7,3)$ in the Rician channel. Both the perfect and the imperfect channel states are assumed at the receiver. The channel uncertainty is characterized by $\gamma^2 = 0.01$ and let $T$ be the cardinality of the candidate set $\mathcal{T}$. From Fig. $\ref{Compr_SVD_versus_ML}$, we observe that: First, the BER performance of the suboptimal detector approaches that of the ML detection as the value of the set cardinality $T$ increasing. This is because the antenna group indices are more likely to be involved in the candidate set for a larger value of $T$; Second, the proposed suboptimal detector with $T = 3$ is sufficiently good to achieve a near-optimal error performance under both perfect and imperfect channel estimation. Thus, the proposed detector achieves a satisfactory tradeoff between complexity and reliability.

\section{Conclusions}
We proposed a SCM scheme to enhance both the reliability and the efficiency of a multi-antenna communication system. By encoding information in the spatial domain, the detection probability of the antenna activation patterns can be directly enhanced, thus the system reliability is significantly improved. The proposed SCM also benefits from the high capacity, the flexible tradeoff between efficiency and reliability, and the compatibility. We show that the proposed SCM outperforms the conventional GSM in various channel scenarios. Furthermore, both the optimal and the suboptimal detectors are formulated for the SCM and the suboptimal alternative strikes a flexible balance between the computational complexity and the reliability. The ramification of this paper is that the spatial coding is a new perspective to improve the performance of the multi-antenna communication systems, especially in the context of large-scale antenna arrays.

\end{document}